\documentclass[11pt]{article}
\usepackage[utf8]{inputenc}
\usepackage{amsmath}
\usepackage{lipsum} 
\usepackage{fullpage}
\usepackage{mathrsfs}
\usepackage{siunitx}
\usepackage{svg}
\usepackage{xr}
\usepackage{amsthm,cite,url,graphicx,booktabs,lipsum,color,bm,caption,subcaption,soul}
\usepackage{pifont,tikz,paralist,multirow,amssymb}
\usepackage{xifthen}
\usepackage{enumerate}
\usepackage{titlesec}
\usepackage{natbib} 
\usepackage{color}
\usepackage{tikz}
\usepackage{threeparttable}
\usepackage{hyperref}
\hypersetup{
    colorlinks=true,
    linkcolor=blue,
    urlcolor=blue,
    citecolor=blue
}
\usetikzlibrary{positioning}
\newtheorem{theorem}{Theorem}
\newtheorem{lemma}{Lemma}

\newtheorem{proposition}{Proposition}
\theoremstyle{definition}

\newtheorem{remark}{Remark}

\newtheorem{assumption}{Assumption}

\newcommand{\red}{\color{red}}

\allowdisplaybreaks

\def\pr{\textnormal{pr}}
\def\AT{{ss}}
\def\CO{{s\bar{s}}}
\def\NT{{\bar{s}\bar{s}}}
\def\DE{{01}}  
  
\makeatletter
\newcommand*{\indep}{%
 \mathbin{%
  \mathpalette{\@indep}{}%
 }%
}
\newcommand*{\nindep}{%
 \mathbin{
  \mathpalette{\@indep}{\not}
 }%
}
\newcommand*{\@indep}[2]{%
 \sbox0{$#1\perp\m@th$}
 \sbox2{$#1=$}
 \sbox4{$#1\vcenter{}$}
 \rlap{\copy0}
 \dimen@=\dimexpr\ht2-\ht4-.2pt\relax
 \kern\dimen@
 {#2}%
 \kern\dimen@
 \copy0 
}
\usepackage{xr}
\def\Pr{\textnormal{pr}}

\def\T{{ \mathrm{\scriptscriptstyle T} }}

\def\AT{{11}}
\def\CO{{01}}
\def\DE{{10}}
\def\NT{{00}}
\def\LL{{11}}

\def\DD{{00}}
\def\PP{{11}}
\def\PN{{01}}
\def\NP{{10}}
\def\NN{{00}}
\usepackage{ulem}

\def\E{\mathbb{E}}  
\usepackage{xr}
\makeatletter
\newcommand*{\addFileDependency}[1]{
  \typeout{(#1)}
  \@addtofilelist{#1}
  \IfFileExists{#1}{}{\typeout{No file #1.}}
}
\makeatother

 \usepackage{tikz}
\usetikzlibrary{shapes.geometric, arrows.meta, positioning, calc}

\begin{document}\hypersetup{linkcolor=black}
\title{\bf \huge {Pseudo-}strata learning via maximizing misclassification reward}   
   
 			\author{Shanshan Luo\textsuperscript{1},    Peng Wu\textsuperscript{1*},   Zhi Geng\textsuperscript{1,2}  
       \\\\	\textsuperscript{1}School of Mathematics and Statistics, Beijing Technology and Business University \\\\ 
  	\textsuperscript{2}School of Mathematical Sciences, Peking University}

\date{} 
\maketitle  
\hypersetup{linkcolor=blue}

\begin{abstract} 
Online advertising aims to increase user engagement and maximize revenue, but users respond heterogeneously to ad exposure. Some users purchase only when exposed to ads, while others purchase regardless of exposure, and still others never purchase. 
This heterogeneity can be characterized by latent response types, commonly referred to as principal strata, defined by users' joint potential outcomes under exposure and non-exposure. However, users' true strata are unobserved, making direct analysis infeasible.   
In this article, instead of learning the true strata, we propose a novel approach that learns users' \emph{pseudo-strata} by leveraging information from an outcome (revenue) observed after the response (purchase). We construct pseudo-strata to classify users and introduce misclassification rewards to quantify the expected revenue gain of pseudo-strata-based policies relative to true strata.    
Within a Bayesian classification framework, we learn the pseudo-strata by optimizing the expected revenue. 
To implement these procedures, we introduce identification assumptions and estimation methods, and establish their large-sample properties.   
Simulation studies show that the proposed method achieves more accurate strata classification and substantially higher revenue than baselines. We further illustrate the method using a large-scale industrial dataset from the Criteo Predictive Search Platform. 
\end{abstract}


\section{Introduction}





\subsection{Background and Main Contributions}

In ad recommendation settings, when users are exposed to an ad, they may exhibit diverse purchase behaviors (responses), resulting in varying levels of revenue (outcomes). 
 To characterize such individual heterogeneity, we introduce a latent variable, referred to as principal stratification (strata) \citep{Frangakis:2002}, which is defined by the joint values of potential response variables.  
 This latent variable typically takes four possible values, representing four distinct user types: always buyers, who purchase regardless of ad exposure; never buyers, who never purchase under any condition; persuadable buyers, who purchase only when exposed to the ad; and discouraged buyers, who purchase only when not exposed to the ad.  
Inferring principal strata enables at least three key applications with significant theoretical and practical implications: 

\begin{itemize} 
\item[(i)] \emph{Causal attribution analysis}. 
Knowledge of principal strata enables causal attribution by exploring users who are persuaded, discouraged, or unaffected by ad exposure~ \citep{pearl2009causality,KUROKIandCAI,shingaki2021identification}.  
    
\item[(ii)] \emph{Optimal ad recommendation policy}.  
Knowledge of principal strata enables the development of optimal ad recommendation policies that maximize revenue under resource constraints. 
Platforms can expose ads to persuadable buyers, avoid discouraged buyers, and exclude always and never buyers~\citep{ben2024policy,li2023trustworthy, Wu-etal2025}.

\item[(iii)] \emph{{Principal causal effects}}. 
Naive revenue comparisons between purchasers with and without ad exposure lack causal interpretability due to latent strata heterogeneity. 
     Valid causal inference requires conditioning on latent principal strata~\citep{Ding:2011,Jiang2016,Wang2017iv,jiang2021identification}. 
 \end{itemize}

Despite these important applications, the true  principal strata are unobserved, rendering direct learning infeasible. 
In this article, instead of learning the true strata, we propose a novel approach that learns users' {pseudo-strata} by leveraging information from an outcome variable (revenue) observed after the response (purchase).
Specifically, we make four main contributions: 
\begin{itemize}
    \item[(i)] We propose a novel conceptual framework for learning {pseudo-strata}. The framework proceeds in three steps. First, we construct pseudo-strata to classify users and define the corresponding optimal  policies. Second, we introduce misclassification rewards to quantify the expected revenue of each policy conditional on users' true {principal strata}. Finally, we define the optimal {pseudo-strata} and learn them by maximizing the expected revenue.
\item[(ii)]  We establish identification results under a new set of  model separation conditions, which disentangle treatment effects across latent response strata, and further ensure that both the misclassification reward and the {pseudo-strata} are identifiable from observed data. 
\item[(iii)]  We develop flexible estimation and inference procedures that accommodate both parametric and semiparametric specifications, and establish their large-sample properties, including consistency and convergence rates.  
\item[(iv)]  Simulation studies show that the proposed method achieves more accurate classification and substantially higher policy revenue than baselines. 
An empirical analysis on a large-scale industrial dataset from the Criteo Predictive Search platform further demonstrates its practical effectiveness, yielding significantly higher revenue. 
\end{itemize}

\subsection{Related Work} 
Instead of learning the true principal strata, much of the existing literature has focused on deriving bounds for the strata probability distribution under minimal assumptions. For example, 
\citet{tian2000probabilities} derived sharp bounds for strata probabilities under the strong ignorability assumption, while \citet{KUROKIandCAI} further tightened these bounds by incorporating covariate information. \citet{dawid2017probability} showed that introducing auxiliary information through a mediator can narrow these bounds, and \citet{Wu-etal-2025-Harm} proposed tighter bounds by imposing constraints on the correlation coefficient between potential outcomes.

Nevertheless, these bounds are often too wide to support precise causal inference; subsequent research has explored conditions under which point identification of strata probabilities becomes achievable. 
\citet{tian2000probabilities} presented identification results under the monotonicity assumption, provided that the marginal causal effects themselves are identifiable. Similarly, \citet{pearl2009causality} derived explicit functional relationships linking causal effects to strata probabilities.  Within the framework of natural direct and indirect effects \citep{pearl2022direct}, \citet{robins2010alternative} demonstrated that under the assumption of no unmeasured confounding, strata probabilities are identifiable if either (i) the two potential response variables are independent, or (ii) one is a deterministic function of the other.  \citet{shingaki2021identification} proposed two classes of identification conditions: one incorporates a single proxy variable together with causal risks, while the other relies on two proxy variables without using causal risks. 
  \citet{Zhang:2009} established that when the outcome variable follows a mixture of normal distributions, the proportions of  strata are also identifiable. Alternative stochastic monotonicity assumptions for two potential purchase outcomes were used by \citet{roy2008modeling} and \citet{lee2010weight} to establish the identification of strata probabilities.

Unlike existing work that focuses on estimating principal strata probabilities (or their bounds), our approach aims to learn pseudo-strata that maximize expected revenue while accounting for potential misclassification of principal strata. This aligns naturally with practical objectives in causal inference, such as quantifying population heterogeneity, characterizing strata-specific causal effects, and designing optimal treatment policies to maximize revenue.

\section{Notation and Setup}
\label{sec:framework}
Assume that $n$ individuals are independently sampled from an infinite superpopulation. 
 For each individual $i$, let $Z_i$ denote the treatment indicator describing whether the individual is exposed to an ad, $X_i$ denote a vector of $m$-dimensional covariates such as demographic characteristics, $S_i$ denote the response status representing whether the individual would purchase the product, and $Y_i$ denote the revenue corresponding to the purchase behavior. Using the potential outcomes framework, each individual has two potential response statuses: $S_{1i}$ if treated (ad exposure) and $S_{0i}$ if untreated (control, ad non-exposure). 
 Similarly,  Similarly, each individual has two potential revenues: $Y_{1i}$, corresponding to ad exposure, and $Y_{0i}$, corresponding to non-exposure. 
Under the consistency assumption, 
$Y_i = Z_i Y_{1i} + (1 - Z_i) Y_{0i}$ and 
$S_i = Z_i S_{1i} + (1 - Z_i) S_{0i}$.  
 For notational simplicity, we  suppress the subscript $i$ in the subsequent exposition.

In practice, some individuals may view an ad but choose not to purchase, while others may purchase without being exposed to the ad. To capture such heterogeneity in  behavior, we adopt the \textit{{type}s} framework \citep{Angrist:1996} or the \textit{principal stratification} framework \citep{Frangakis:2002,Rubin2006} under which individuals are categorized into four latent strata: 
 always buyers ($G = \LL$, with $S_0 = S_1 = 1$) purchase regardless of treatment; never buyers ($G = \DD$, with $S_0 = S_1 = 0$) never purchase; persuadable buyers ($G = 01$, with $S_0 = 0, S_1 = 1$) purchase only when treated; and discouraged buyers ($G = 10$, with $S_0 = 1, S_1 = 0$) purchase only when untreated.  Let $\pi_{s_0s_1}(X):=\pr(G=s_0s_1\mid X)$ denote the conditional probability that an individual with covariates $X$ belongs to the type $G=s_0s_1$.
 
 Although the {four strata} provide a useful conceptual framework, each individual's true {type} $G$ remains unobserved. Learning this latent variable is critical for three key applications:   
(i) principal causal effects, such as understanding treatment effect heterogeneity across strata \citep{Ding:2011,Jiang2016,Wang2017iv};
(ii)  causal attribution, such as determining which ads drive conversions \citep{KUROKIandCAI,shingaki2021identification}; and
(iii) personalized recommendation policy, such as targeting high-value users to maximize revenue  \citep{ben2024policy,li2023trustworthy,Wu-etal-2025-Harm}. 
In this article, we propose a novel classification-based approach that learns individuals' principal strata while maximizing expected revenue under misclassification uncertainty. Before introducing the proposed method, we discuss a commonly used approach and reveal its limitation in our setting, see Remark \ref{remark1} below. 

\begin{remark}  \label{remark1}
An intuitive approach, referred to as the posterior mode rule, infers an individual's strata by computing the probabilities of belonging to each of the {four strata} (\(\PP, \PN, \NP, \NN\)) and assigning the individual to the stratum with the highest probability (see, e.g., \citealp{tian2000probabilities, KUROKIandCAI, dawid2017probability, pearl2009causality, Zhang:2009, shingaki2021identification}); also see \eqref{eq:post-ps} in Section \ref{sec:policy-studies} for details. 
However, these probabilities are typically derived from $X$ and $S$ alone, without incorporating the information contained in $Y$. Such a classification approach often yields less accurate strata, resulting in suboptimal recommendation policies and substantial economic losses, as demonstrated in our simulation studies and real-world applications, where this method is included as a baseline for numerical comparison. 
\end{remark}

Let ${\mathcal{L}}_{z,s_0s_1}(X) = \mathbb{E}(Y_z \mid G = s_0s_1, X)$ denote the conditional expectation of the potential outcome $Y_z$ given the true principal stratum $G=s_0s_1$ and covariates $X$. In the context of personalized treatment assignment, ${\mathcal{L}}_{z,s_0s_1}(X)$ represents the expected revenue (utility) from assigning treatment $z$ to an individual in stratum $s_0s_1$ with covariates $X$.
A key feature of our setting is that the post-treatment outcome (revenue) is only observed for individuals who complete the intermediate response (purchase). This structural constraint naturally implies   
$ {\mathcal{L}}_{0,01}(X) = {\mathcal{L}}_{1,10}(X) = {\mathcal{L}}_{0,00}(X) = {\mathcal{L}}_{1,00}(X) = 0,$ 
which states that individuals who do not purchase under treatment $z$ generate zero revenue (utility). 
We define the marginal expected utility for stratum $G=s_0s_1$ under treatment $z$ as $\mathcal{L}_{z,s_0s_1} := \mathbb{E}(Y_z \mid G = s_0s_1).$
\section{Statistical Framework for Pseudo-Strata Learning}
\label{sec: Preliminaries}
\subsection{Optimal Treatment Policy Based on Strata}
\label{sec:optimal-policy}
In this section, we assume that the causal estimands $\pi_{s_0s_1}(X) $ and ${\mathcal{L}}_{z,s_0s_1}(X)$ are known, to illustrate the main ideas of the proposed framework.  
Given covariates $X$, let $c_1(X)$ and $c_0(X)$ denote the known costs of treatment and control, respectively. 
In an ideal scenario where the true principal stratification $G$ were known, the optimal treatment policy $\rho(G, X)$ is defined as follows: 
\begin{itemize}
\item  {Always buyers ($G = \PP$):} For individuals who purchase regardless 
of treatment, the decision rule 
$\rho(\PP,X) = \mathbb{I}\{{\mathcal{L}}_{1,\PP}(X) - c_1(X) \geq 
{\mathcal{L}}_{0,\PP}(X) - c_0(X)\}$ compares net  profits under both conditions. 
Treatment ($Z = 1$) is recommended if its net profit exceeds that of control. 
otherwise, control ($Z = 0$) is recommended.

\item  {Persuadable buyers ($G = \PN$):} For individuals who only 
purchase when treated, the {policy} 
$\rho(\PN,X) = \mathbb{I}\{{\mathcal{L}}_{1,\PN}(X) - c_1(X) \geq -c_0(X)\}$ 
compares the net profit under treatment with the cost of control. Treatment 
is recommended if the net {profit} under treatment exceeds the cost of control; 
Otherwise, control is recommended.

\item  {Discouraged buyers ($G = \NP$):} For individuals who only 
purchase when untreated, the policy $\rho(\NP,X) = \mathbb{I}\{-c_1(X) \geq {\mathcal{L}}_{0,\NP}(X) - c_0(X) \}$
compares the net profit under control with the cost of treatment. Control 
is recommended if its net {profit} exceeds the cost of treatment; otherwise, 
treatment is recommended. 

\item  {Never buyers ($G = \NN$):} For individuals who never purchase, 
the {policy} $\rho(\NN,X) = \mathbb{I}\{-c_1(X) \geq -c_0(X)\}$ 
(equivalently, $c_1(X) \leq c_0(X)$) depends solely on minimizing cost. 
Treatment is recommended if $c_1(X) \leq c_0(X)$; otherwise, control is 
recommended.
\end{itemize}

The above policy $\rho(s_0s_1, X)$ depends on the true principal stratum $G$ for each individual. However, $G$ is unobserved in practice. To proceed, in the next subsection, we construct a pseudo-stratum $\tilde{G}$ to approximate each individual's true stratum $G$. The pseudo-stratum $\tilde{G}$ is a classifier that maps each individual to one of the four strata $\{00, 01, 10, 11\}$ based on their observed characteristics $X$. We suppress $X$ in the notation for $\tilde{G}$ to emphasize that it serves as an approximation to the true $G$. Once each individual is assigned a pseudo-stratum $\tilde{G} = \tilde{s}_0 \tilde{s}_1$, we apply the corresponding treatment policy $\rho(\tilde{s}_0 \tilde{s}_1, X)$ for tailored ad recommendations.

\begin{remark}\label{remark2} 
 {In terms of policy learning alone}, an alternative approach to obtaining the optimal policy is to optimize the overall revenue directly. Specifically, for a given policy $d(X)$, when $\pi_{s_0s_1}(X)$ and $\mathcal{L}_{z,s_0s_1}(X)$ are known, the overall revenue (or value function) $V(d) $ is defined as 
\begin{equation}
    \label{eq:policy-value}
    \begin{aligned}
   & V(d) = \\
 &    \sum_{  {s}_0 ,{s}_1\in\{0,1\} }\E    \left( \mathbb{I}(G=  s_0s_1)  \Big[   {d}( X)\left\{  \mathcal{L}_{1, s_0s_1}(X) -   c_1(X) \right\} + \{1- {d}( X)\} \left \{\mathcal{L}_{0, s_0s_1}(X)-   c_0(X) \right\}\Big ]    \right )\\&= \sum_{  {s}_0 ,{s}_1\in\{0,1\} } \E \left(   \pi_{s_0s_1}(X)\Big[  {d}( X) \left\{  \mathcal{L}_{1, s_0s_1}(X) -   c_1(X) \right\} + \{1- {d}( X)\} \left \{ \mathcal{L}_{0, s_0s_1}(X)-   c_0 (X)  \right\} \Big ]    \right ).    
 \end{aligned}
\end{equation}
The optimal policy is the one that maximizes $V(d)$ (possibly within a pre-specified policy class). We compare our proposed method with this direct optimization approach in both simulation studies and real data analysis, where our method demonstrates superior performance.

\end{remark}

\subsection{Misclassification Probability and   Reward}
In this subsection, we consider the misclassification probability and the reward induced by the pseudo-strata $\tilde{G}$.  
For any given pseudo-strata $\tilde{G}$, it induces a partition of the covariate space $\mathcal{X}$: $ \tilde{D}_{\tilde{s}_0\tilde{s}_1} = \{ X \in \mathcal{X} : \tilde{G} = \tilde{s}_0\tilde{s}_1\},  $ where $\tilde{D}_{00} \cup \tilde{D}_{01} \cup \tilde{D}_{10} \cup \tilde{D}_{11} = \mathcal{X}$ and the sets are mutually disjoint. 
  Let $\Pr(\tilde{G} = \tilde{s}_0 \tilde{s}_1 \mid G = s_0 s_1)$ denote the probability that individuals with true principal stratum $G = s_0 s_1$ are classified into pseudo-stratum $\tilde{G} = \tilde{s}_0 \tilde{s}_1$. We express this \textit{misclassification probability} as:
$$
\Pr(\tilde{G} = \tilde{s}_0\tilde{s}_1 \mid G = s_0s_1) = \Pr(X\in {\tilde{D}_{\tilde{s}_0\tilde{s}_1}} \mid G = s_0s_1) = \int_{\tilde{D}_{\tilde{s}_0\tilde{s}_1}} f(x \mid G = s_0s_1)\, dx,
$$
where $f(x \mid G = s_0s_1)$ denotes the conditional density of covariates $X$ given true principal stratum $G = s_0s_1$.

Figure \ref{fig:misclassification-G11} provides an illustration of the misclassification probabilities $\Pr(\tilde{G}=\tilde{s}_0\tilde{s}_1 \mid G=11)$ for a given true principal stratum $G=11$, assuming that the covariate space $\mathcal{X}$ is partitioned into four regions, $\tilde{D}_{00}$, $\tilde{D}_{01}$, $\tilde{D}_{10}$, and $\tilde{D}_{11}$, corresponding to the four pseudo-strata.  
The conditional density $f(x \mid G=11)$ illustrates how individuals with true principal stratum $G=11$ are distributed across the covariate space, with the shaded area under the curve in each region representing the corresponding misclassification probability $\Pr(\tilde{G}=\tilde{s}_0\tilde{s}_1 \mid G=11)$.
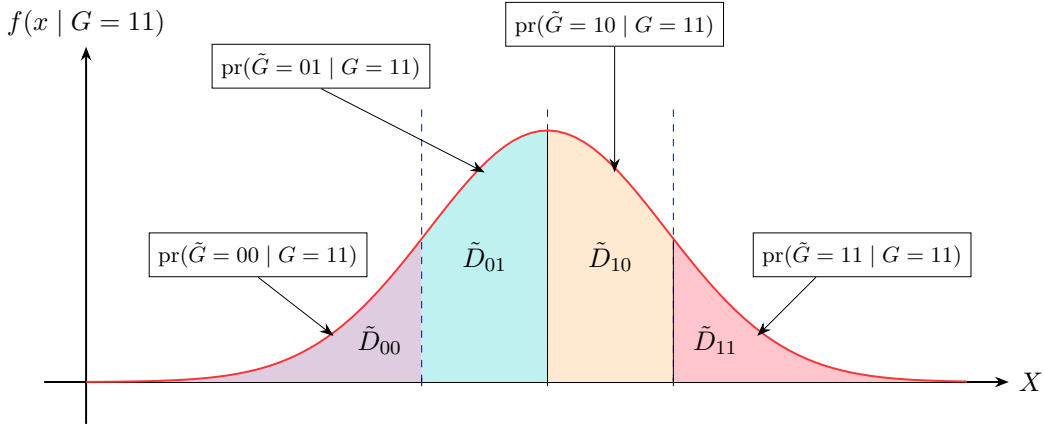
\begin{figure}
\centering
\definecolor{color11}{RGB}{255,182,193}   
\definecolor{color10}{RGB}{255,228,196}   
\definecolor{color01}{RGB}{175,238,238}   
\definecolor{color00}{RGB}{215.5,191,216}   
\resizebox{0.84905\textwidth}{!}{\begin{tikzpicture}[scale=1.2, >={Stealth[length=2mm]}]
\draw[->,thick] (-0.5,0) -- (11,0) node[right] {$X$};
\draw[->,thick] (0,-0.5) -- (0,4) node[above] {$f(x \mid G=11)$};
\fill[color00!80,domain=0:4,samples=50] 
  (0,0) -- plot (\x,{3*exp(-0.5*(\x-5.5)^2/2)}) -- (4,0) -- cycle;
\fill[color01!80,domain=4:5.5,samples=50] 
  (4,0) -- plot (\x,{3*exp(-0.5*(\x-5.5)^2/2)}) -- (5.5,0) -- cycle;
\fill[color10!80,domain=5.5:7,samples=50] 
  (5.5,0) -- plot (\x,{3*exp(-0.5*(\x-5.5)^2/2)}) -- (7,0) -- cycle;
\fill[color11!80,domain=7:10.5,samples=50] 
  (7,0) -- plot (\x,{3*exp(-0.5*(\x-5.5)^2/2)}) -- (10.5,0) -- cycle;
\draw[dashed,blue] (4,0) -- (4,3.25);
\draw[dashed,blue] (5.5,0) -- (5.5,3.25);
\draw[dashed,blue] (7,0) -- (7,3.25);
\node at (3.5,0.5) {$\tilde{D}_{{0}{0}}$};
\node at (4.75,1.5) {$\tilde{D}_{{0}{1}}$};
\node at (6.25,1.5) {$\tilde{D}_{{1}{0}}$};
\node at (7.5,0.5) {$\tilde{D}_{11}$};
\draw (5.5,0) -- (5.5,{3*exp(-0.5*(5.5-5.5)^2/2)});
\draw (7,0) -- (7,{3*exp(-0.5*(7-5.5)^2/2)});
\draw[thick,red!80,domain=0:10.5,samples=100] plot (\x,{3*exp(-0.5*(\x-5.5)^2/2)});
\draw[->] (6.3,4) -- (6.3,2.5);
\node[above,align=center,draw,fill=white] at (6.3,4) {\footnotesize $\pr(\tilde{G}=10\mid G=11)$};
\draw[->] (2.8,3.5) -- (4.75,2.5);
\node[above,align=center,draw,fill=white] at (2.8,3.5) {\footnotesize $\pr(\tilde{G}=01\mid G=11)$};
\draw[->] (2.05,1.25) -- (2.95,0.55);
\node[above,align=center,draw,fill=white] at (2.05,1.25) {\footnotesize $\pr(\tilde{G}=00\mid G=11)$};
\draw[->] (2.95+5.75,1.25) -- (2.25+5.75,0.55);
\node[above,align=center,draw,fill=white] at (2.25+7,1.25) {\footnotesize $\pr(\tilde{G}=11\mid G=11)$};
 
\node[red!80,above right] at (8.5,2.8) { };
\foreach \x in {4,5.5,7} {
    \draw (\x,0.05) -- (\x,-0.05);
}
\end{tikzpicture}}\caption{Illustration of misclassification probabilities  $\pr(\tilde{G}=\tilde{s}_0\tilde{s}_1\mid G=11)$ for a given true principal stratum $G=11$. }

    \label{fig:misclassification-G11}
\end{figure}

In addition to accounting for the probability of misclassification, it is also important to assess the potential revenue when a misclassified individual receives a treatment inconsistent with the true stratification.  
Specifically, suppose an individual's true principal stratum is $G = s_0s_1$, but they are misclassified as $\tilde{G} = \tilde{s}_0\tilde{s}_1$ and consequently receive the policy $\rho(\tilde{s}_0 \tilde{s}_1, X)$ defined in  Section~\ref{sec:optimal-policy}.  The individual's actual outcomes will follow the response pattern of their true principal stratum ${G} =s_0s_1$, but the treatment policy is determined by the misclassified stratum $\tilde{G} = \tilde{s}_0\tilde{s}_1$, leading to a potentially suboptimal revenue realization. 
 We define the expected revenue under such misclassification, referred to as the \textit{conditional misclassification reward}, as follows: 
\begin{equation}
\label{eq:misclassification-reward}
    \begin{aligned}
\mathcal{R}( \tilde{s}_0\tilde{s}_1 \mid   s_0s_1) 
={}& \mathbb{E}\left[\{{\mathcal{L}}_{1,s_0s_1}(X) -{c_1(X)}\}  \rho( {\tilde{s}_0\tilde{s}_1},X) \mid G = s_0s_1\right]  \\
{}& + \mathbb{E} \left[ \{{\mathcal{L}}_{0,s_0s_1}(X)-{c_0(X)}\} \{1 - \rho( {\tilde{s}_0\tilde{s}_1},X)\} \mid G = s_0s_1\right].
\end{aligned}
\end{equation}
The quantity $\mathcal{R}( \tilde{s}_0\tilde{s}_1 \mid s_0s_1)$ includes three key components: 
(i) ${\mathcal{L}}_{z,s_0s_1}(X) - c_z(X)$ represents the net profit from treatment $Z=z$ for true principal stratum $G=s_0 s_1$; (ii) $\rho(\tilde{s}_0\tilde{s}_1, X)$ determines the treatment policy based on the pseudo-strata $\tilde{G}=\tilde{s}_0 \tilde{s}_1$; and (iii) the expectation is conditional on the true principal stratum $s_0 s_1$, ensuring evaluation relative to actual response behavior. 

When $\tilde{G} = G=s_0s_1$, $\mathcal{R}(s_0s_1 \mid s_0s_1)$ attains the maximum expected revenue; this  implicitly suggests that the outcome $Y$ may provide information on the learning principal-strata labels $G$.  
When $\tilde{G} \neq G$, $\mathcal{R}(s_0s_1 \mid s_0s_1)$ quantifies the revenue under misclassification, providing a revenue-based criterion for evaluating the average performance of pseudo-strata $\tilde G$.   

\subsection{Determining Pseudo-Strata by Maximizing the Average Misclassification Reward} 
\begin{figure}[t]
\centering
  \resizebox{0.948605\textwidth}{!}{
\begin{tikzpicture}[
    node distance=0.35cm and 1.5cm,
    box/.style={rectangle, draw, thick, minimum width=3cm, minimum height=0.8cm, align=center},
    roundbox/.style={rectangle, rounded corners=5pt, draw, thick, minimum width=1.8cm, minimum height=0.7cm, align=center},
    predbox/.style={rectangle, rounded corners=5pt, draw, thick, minimum width=3.2cm, minimum height=0.70cm, align=center, font=\small},
    probbox/.style={rectangle, rounded corners=5pt, draw, thick, minimum width=3.5cm, minimum height=0.70cm, align=center, font=\small},
    arrow/.style={-Stealth, thick}
]
\definecolor{color11}{RGB}{255,182,193}   
\definecolor{color10}{RGB}{255,228,196}   
\definecolor{color01}{RGB}{175,238,238}   
\definecolor{color00}{RGB}{216,191,216}   

\node[box] (title1) at (0, 9) {\textbf{True stratum}};
\node[box] (title2) at (5.5, 9) {\textbf{Conditional}\\\textbf{misclassification reward}};
\node[box] (title3) at (11.3, 9) {\textbf{Misclassification}\\\textbf{probability}};

\node[box] (avg) at (-3.8, 3) {\textbf{Average}\\\textbf{misclassification}\\\textbf{reward}};

\node[roundbox, fill=color11] (g11) at (0, 6.3) {$G = 11$};
\node[above=0.05cm of g11, font=\small] {$\Pr(G = 11)$};

\node[roundbox, fill=color10] (g10) at (0, 3.5) {$G = 10$};
\node[above=0.05cm of g10, font=\small] {$\Pr(G = 10)$};

\node[roundbox, fill=color01] (g01) at (0, 0.7) {$G = 01$};
\node[below=0.05cm of g01, font=\small] {$\Pr(G = 01)$};

\node[roundbox, fill=color00] (g00) at (0, -2.1) {$G = 00$};
\node[below=0.05cm of g00, font=\small] {$\Pr(G = 00)$};

\node[predbox, fill=color11] (p11_11) at (5.5, 7.65) {$\widetilde{G} = 11 : \mathcal{R}(11 \mid 11)$};
\node[predbox, fill=color11] (p11_10) at (5.5, 6.85) {$\widetilde{G} = 10 : \mathcal{R}(10 \mid 11)$};
\node[predbox, fill=color11] (p11_01) at (5.5, 6.05) {$\widetilde{G} = 01 : \mathcal{R}(01 \mid 11)$};
\node[predbox, fill=color11] (p11_00) at (5.5, 5.25) {$\widetilde{G} = 00 : \mathcal{R}(00 \mid 11)$};

\node[predbox, fill=color10] (p10_11) at (5.5, 4.25) {$\widetilde{G} = 11 : \mathcal{R}(11 \mid 10)$};
\node[predbox, fill=color10] (p10_10) at (5.5, 3.45) {$\widetilde{G} = 10 : \mathcal{R}(10 \mid 10)$};
\node[predbox, fill=color10] (p10_01) at (5.5, 2.65) {$\widetilde{G} = 01 : \mathcal{R}(01 \mid 10)$};
\node[predbox, fill=color10] (p10_00) at (5.5, 1.85) {$\widetilde{G} = 00 : \mathcal{R}(00 \mid 10)$};

\node[predbox, fill=color01] (p01_11) at (5.5, 0.85) {$\widetilde{G} = 11 : \mathcal{R}(11 \mid 01)$};
\node[predbox, fill=color01] (p01_10) at (5.5, 0.05) {$\widetilde{G} = 10 : \mathcal{R}(10 \mid 01)$};
\node[predbox, fill=color01] (p01_01) at (5.5,-0.75) {$\widetilde{G} = 01 : \mathcal{R}(01 \mid 01)$};
\node[predbox, fill=color01] (p01_00) at (5.5, -1.55) {$\widetilde{G} = 00 : \mathcal{R}(00 \mid 01)$};

\node[predbox, fill=color00] (p00_11) at (5.5, -2.55) {$\widetilde{G} = 11 : \mathcal{R}(11 \mid 00)$};
\node[predbox, fill=color00] (p00_10) at (5.5, -3.35) {$\widetilde{G} = 10 : \mathcal{R}(10 \mid 00)$};
\node[predbox, fill=color00] (p00_01) at (5.5, -4.15) {$\widetilde{G} = 01 : \mathcal{R}(01 \mid 00)$};
\node[predbox, fill=color00] (p00_00) at (5.5, -4.95) {$\widetilde{G} = 11 : \mathcal{R}(11 \mid 00)$};

\node[probbox, fill=color11] (pr11_11) at (11.3, 7.65) {$\Pr(\widetilde{G} = 11 \mid G = 11)$};
\node[probbox, fill=color11] (pr11_10) at (11.3, 6.85) {$\Pr(\widetilde{G} = 10 \mid G = 11)$};
\node[probbox, fill=color11] (pr11_01) at (11.3, 6.05) {$\Pr(\widetilde{G} = 01 \mid G = 11)$};
\node[probbox, fill=color11] (pr11_00) at (11.3, 5.25) {$\Pr(\widetilde{G} = 00 \mid G = 11)$};

\node[probbox, fill=color10] (pr10_11) at (11.3, 4.25) {$\Pr(\widetilde{G} = 11 \mid G = 10)$};
\node[probbox, fill=color10] (pr10_10) at (11.3, 3.45) {$\Pr(\widetilde{G} = 10 \mid G = 10)$};
\node[probbox, fill=color10] (pr10_01) at (11.3, 2.65) {$\Pr(\widetilde{G} = 01 \mid G = 10)$};
\node[probbox, fill=color10] (pr10_00) at (11.3, 1.85) {$\Pr(\widetilde{G} = 00 \mid G = 10)$};

\node[probbox, fill=color01] (pr01_11) at (11.3, 0.85) {$\Pr(\widetilde{G} = 11 \mid G = 01)$};
\node[probbox, fill=color01] (pr01_10) at (11.3, 0.05) {$\Pr(\widetilde{G} = 10 \mid G = 01)$};
\node[probbox, fill=color01] (pr01_01) at (11.3, -0.75) {$\Pr(\widetilde{G} = 01 \mid G = 01)$};
\node[probbox, fill=color01] (pr01_00) at (11.3, -1.55) {$\Pr(\widetilde{G} = 00 \mid G = 01)$};

\node[probbox, fill=color00] (pr00_11) at (11.3, -2.55) {$\Pr(\widetilde{G} = 11 \mid G = 00)$};
\node[probbox, fill=color00] (pr00_10) at (11.3, -3.35) {$\Pr(\widetilde{G} = 10 \mid G = 00)$};
\node[probbox, fill=color00] (pr00_01) at (11.3, -4.15) {$\Pr(\widetilde{G} = 01 \mid G = 00)$};
\node[probbox, fill=color00] (pr00_00) at (11.3, -4.95) {$\Pr(\widetilde{G} = 11 \mid G = 00)$};

\draw[arrow] (avg) -- (g11);
\draw[arrow] (avg) -- (g10);
\draw[arrow] (avg) -- (g01);
\draw[arrow] (avg) -- (g00);

\draw[arrow] (g11.east) -- (p11_11.west);
\draw[arrow] (g11.east) -- (p11_10.west);
\draw[arrow] (g11.east) -- (p11_01.west);
\draw[arrow] (g11.east) -- (p11_00.west);

\draw[arrow] (g10.east) -- (p10_11.west);
\draw[arrow] (g10.east) -- (p10_10.west);
\draw[arrow] (g10.east) -- (p10_01.west);
\draw[arrow] (g10.east) -- (p10_00.west);

\draw[arrow] (g01.east) -- (p01_11.west);
\draw[arrow] (g01.east) -- (p01_10.west);
\draw[arrow] (g01.east) -- (p01_01.west);
\draw[arrow] (g01.east) -- (p01_00.west);

\draw[arrow] (g00.east) -- (p00_11.west);
\draw[arrow] (g00.east) -- (p00_10.west);
\draw[arrow] (g00.east) -- (p00_01.west);
\draw[arrow] (g00.east) -- (p00_00.west);

\draw[arrow] (p11_11.east) -- (pr11_11.west);
\draw[arrow] (p11_10.east) -- (pr11_10.west);
\draw[arrow] (p11_01.east) -- (pr11_01.west);
\draw[arrow] (p11_00.east) -- (pr11_00.west);

\draw[arrow] (p10_11.east) -- (pr10_11.west);
\draw[arrow] (p10_10.east) -- (pr10_10.west);
\draw[arrow] (p10_01.east) -- (pr10_01.west);
\draw[arrow] (p10_00.east) -- (pr10_00.west);

\draw[arrow] (p01_11.east) -- (pr01_11.west);
\draw[arrow] (p01_10.east) -- (pr01_10.west);
\draw[arrow] (p01_01.east) -- (pr01_01.west);
\draw[arrow] (p01_00.east) -- (pr01_00.west);

\draw[arrow] (p00_11.east) -- (pr00_11.west);
\draw[arrow] (p00_10.east) -- (pr00_10.west);
\draw[arrow] (p00_01.east) -- (pr00_01.west);
\draw[arrow] (p00_00.east) -- (pr00_00.west);

\draw[dotted, very thick] ($(pr11_11.north west) + (-0.25, 0.25)$) rectangle ($(pr00_00.south east) + (0.25, -0.25)$);

\draw[thick, rounded corners] (-5.85, -5.8) rectangle (13.8, 10.2);

\end{tikzpicture}
}
    \caption{Structural illustration of the average misclassification reward  \eqref{eq:decision-rule} for a given   label $\tilde{G}$.}
    \label{fig:missclassification-reward}
\end{figure}
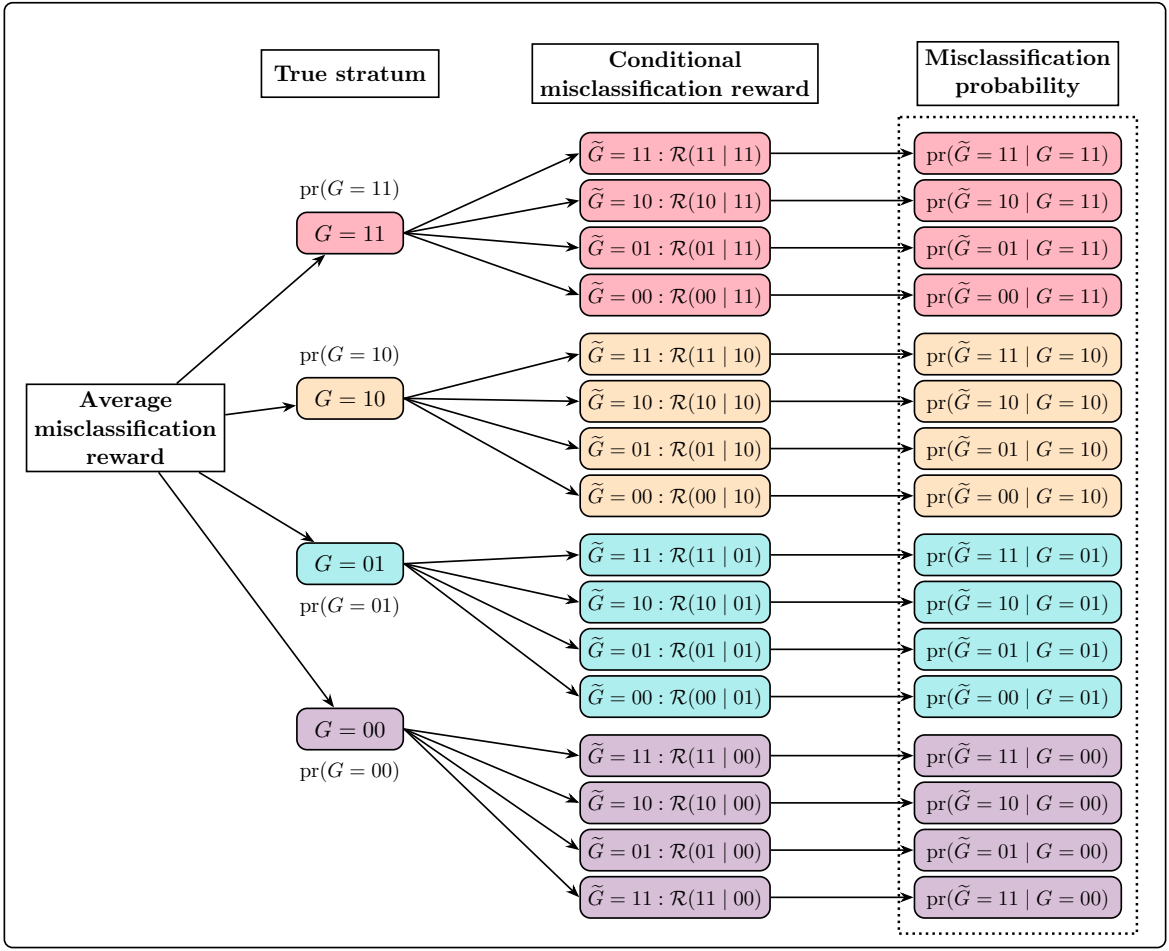

With the concepts of misclassification probabilities and misclassification 
rewards, we introduce the \textit{average misclassification 
reward} for any given pseudo-strata $\tilde{G}$:
\begin{align}
    \label{eq:decision-rule}
    V(\tilde{G}) = \sum_{s_0,s_1\in\{0,1\}} \Pr(G = s_0s_1)
    \sum_{\tilde{s}_0,\tilde{s}_1\in\{0,1\}} \Pr(\tilde{G} = \tilde{s}_0 \tilde{s}_1 
    \mid G = s_0s_1)\mathcal{R}(\tilde{s}_0 \tilde{s}_1 \mid s_0s_1).
\end{align}
The inner sum computes the expected revenue for individuals with true principal stratum $G = s_0s_1$, averaging over all possible pseudo-strata, each weighted by its misclassification probability. The outer sum averages these expected revenues over all strata, weighted by their marginal probabilities $\Pr(G = s_0s_1)$. 
Figure 
\ref{fig:missclassification-reward} illustrates the structure of the average misclassification reward.  The left column shows the four true principal strata $G \in \{00, 01, 10, 11\}$ (corresponding to the outer sum), each weighted by its marginal probability $\Pr(G = s_0s_1)$.  
Given $G = s_0s_1$, the middle column displays all possible pseudo-strata assignments $\tilde{G} = \tilde{s}_0\tilde{s}_1$ (corresponding to the inner sum), each associated with a conditional misclassification reward $\mathcal{R}(\tilde{s}_0\tilde{s}_1 \mid s_0s_1)$. 
The right column shows the corresponding misclassification probabilities $\Pr(\tilde{G} = \tilde{s}_0\tilde{s}_1 \mid G = s_0s_1)$ that weight each reward. The average misclassification reward $V(\tilde{G})$ aggregates across all paths from true  principal strata to pseudo-strata, combining their misclassification probabilities and rewards.


An optimal pseudo-strata ${G}^*$ is said to follow the optimal Bayesian decision rule if it maximizes the average misclassification reward in \eqref{eq:decision-rule}: $
{G}^* = \arg\max_{\tilde{G}} V(\tilde{G}),$ 
where the value function $V(\tilde{G})$ quantifies the expected reward from assigning individuals to pseudo-strata based on their observed covariates. The optimal decision rule induces a partition of the covariate space $\mathcal{X}$ into four mutually exclusive and exhaustive decision regions, $\{D_{\PP}^*, D_{\PN}^*, D_{\NP}^*, D_{\NN}^*\}$, where each region corresponds to one of the four latent strata:
\begin{equation*}
D_{s_0s_1}^* = \left\{x \in \mathcal{X} : {G}^* = s_0s_1\right\}, \quad (s_0, s_1) \in \{0,1\}^2.
\end{equation*}
Individuals falling in region $D_{s_0s_1}^*$ are assigned to the optimal pseudo-stratum ${G}^* = s_0s_1$. 
In the following theorem, we derive the closed-form solution of the optimal pseudo-stratum and characterize its dependence on both the misclassification rewards and the strata probabilities. 

\begin{theorem}
\label{thm:maxim-reward}
The optimal Bayesian decision rule assigns each individual with 
covariates $X$ to the pseudo-strata that maximizes the expected revenue: 
\begin{align}
\label{eq:otr-bayesian}
 {G}^* = \underset{\tilde s_0,\tilde s_1\in\{0,1\}}{\arg\max}~ h_{\tilde s_0\tilde s_1}(X) , 
\end{align} 
where $
h_{\tilde s_0\tilde s_1}(X) = \textstyle\sum_{ {s}_0, {s}_1\in\{0,1\}} 
\mathcal{R}(\tilde{s}_0\tilde{s}_1 \mid s_0s_1) \cdot 
\pi_{ {s}_0 {s}_1}(X).$ 
Equivalently, the optimal decision region $D_{s_0s_1}^\ast$, corresponding to the optimal Bayesian decision rule  $ {G}^*$,  
is given by 
\begin{equation*} 
D_{s_0s_1}^\ast = \left\{ X \in \mathcal{X}: h_{s_0s_1}(X) 
\geq h_{\bar{s}_0\bar{s}_1}(X) ~\text{for all}~  (\bar{s}_0, \bar{s}_1) \in \{0,1\}^2 \right\}.
\end{equation*}
\end{theorem}

 Theorem~\ref{thm:maxim-reward} presents the explicit form of the optimal pseudo-strata that maximizes the expected revenue. For each pair $(s_0, s_1)$, the function $h_{s_0s_1}(X)$ represents the expected revenue for individuals truly belonging to the true principal stratum $ G=s_0s_1$, 
computed as a weighted average of the rewards  $\mathcal{R}(\tilde{s}_0\tilde{s}_1 \mid s_0s_1)$ across all possible  pseudo-strata, with weights being the corresponding  probabilities $\pi_{\tilde{s}_0\tilde{s}_1}(X)$.
By construction, the optimal ${G}^*$ in \eqref{eq:otr-bayesian} depends on the covariates $X$. 
However, since our primary interest lies in how $G^*$ approximates the true $G$ in terms of misclassification rewards, we suppress the explicit dependence on $X$  for simplicity and clarity.

Our proposed approach conceptually resembles the expected cost minimization (ECM) framework in Bayesian classification theory \citep[Chapter 11]{johnson2002applied}, but it differs fundamentally: the standard ECM classifies observed labels, whereas our approach targets unobserved latent principal strata. 
This distinction results in different classification rules and necessitates novel identification conditions. 

The pseudo-strata learned by our approach may not perfectly align with the true latent strata. This stems from the fact that the true strata are defined at the individual level, whereas the optimal Bayesian decision rule operates at the covariate-defined subpopulation level.  
However, they represent a Bayes-optimal solution that balances revenue maximization and classification error in practice. 
Once the optimal Bayesian decision rule $G^*$ is determined, we can utilize the learned pseudo-strata for multiple applications, see Sections \ref{sec:est-proc}, \ref{sec:experiments}, and \ref{sec:application} for details.

\section{Identification Results}
\label{sec:Identification} 

In this section, we investigate the identification results of the causal estimands $\pi_{s_0s_1}(X)$, ${\mathcal{L}}_{z,s_0s_1}(X)$, and $\mathcal{R}(\tilde{s}_0\tilde{s}_1\mid s_0s_1)$ discussed in the previous section. We begin by introducing several commonly used assumptions in causal inference \citep{rosenbaum1983central}. 

\begin{assumption} [Ignorability]
\label{assumption:RCT}
$ Z\indep (Y_0,Y_1,S_0,S_1)\mid X $.
\end{assumption}  

Assumption~\ref{assumption:RCT} states that, conditional on  covariates $X$, there is no unmeasured confounding between the treatment $Z$ and the purchase status $S$, nor between the treatment $Z$ and the revenue $Y$. 
In practice, collecting a comprehensive set of covariates is crucial for improving the applicability of this assumption. It is worth noting that Assumption~\ref{assumption:RCT} does not preclude the presence of unmeasured confounding between $S$ and  $Y$.

To address the identifiability of $\pi_{s_0s_1}(X)$, we adopt an odds ratio parameterization to characterize the dependence structure between potential purchase statuses $S_0$ and $S_1$ \citep{zhang2013assessing,ciocuanea2025sensitivity,tong2025semiparametric, Wu-etal-2025-Harm}: $$\theta(X) = \frac{\pi_{00}(X)\pi_{11}(X)}{\pi_{10}(X)\pi_{01}(X)}.$$
This parameter captures the conditional association between $S_0$ and $S_1$ given $X$. We consider the  assumption for the odds ratio below. 
\begin{assumption}
\label{assump:odds}
The odds ratio function $ \theta(  X) $ is known.
\end{assumption}

Assumption \ref{assump:odds} links the two potential purchase statuses $S_0$ and  $S_1 $. Although $\theta(X) $ is not identifiable from the observed data, we recommend treating it as a sensitivity parameter by varying it over an informative range rather than fixing it to a specific value. The choice of this range can be guided by domain knowledge, with common specifications including: 
(i) Conditional independence ($\theta(X) = 1$): it corresponds to $S_1 \indep S_0 \mid X$, implying that the two potential outcomes are independent given covariates. 
(ii) Monotonicity ($\theta(X) \to +\infty$): 
$S_1 \geq S_0$ or $S_0 \geq S_1$, where the former rules out discouraged buyers by imposing $\pi_{10}(X) = 0$, and the latter rules out persuadable buyers by imposing $\pi_{01}(X) = 0$. 
(iii) Positive association ($\theta(X) > 1$): potential outcomes $S_0$ and $S_1$ are positively associated, meaning individuals likely to purchase under control tend to also be likely to purchase under treatment. 
(iv) Negative association ($\theta(X) < 1$): potential outcomes $S_0$ and $S_1$ are negatively associated, meaning individuals likely to purchase under control tend to be less likely to purchase under treatment.

Although  Assumptions \ref{assumption:RCT} and \ref{assump:odds}  ensure  the identifiability of $\pi_{s_0s_1}(X)$ \citep{zhang2013assessing,ciocuanea2025sensitivity,tong2025semiparametric, Wu-etal-2025-Harm}, it does not suffice for identifying ${\mathcal{L}}_{z,s_0s_1}(X)$. 
To facilitate identification, we decompose the covariates as $X = (C, A)$ without loss of generality, where this partition enables us to exploit covariates that exhibit differential relationships with the strata $G$. Building on this decomposition, we propose a separable model assumption that allows the effects of $C$ and $A$ on $Y$ to be modeled separately. 
{
\begin{assumption}[Additive model]

\label{assumption:NoInteraction} 
Assume that  
${\mathcal{L}}_{0,\PP}(X)$, ${\mathcal{L}}_{1,\PP}(X)$, ${\mathcal{L}}_{1,\PN}(X)$, and ${\mathcal{L}}_{0,\NP}(X)$ follow additive models of the form:
\[
{\mathcal{L}}_{z,s_0s_1}(X) = \mu_{z,s_0s_1,a}^\T q(A) + {\mathcal{L}}_{z,s_0s_1}^\ast(C),
\]
where $q(A) = \{q_1(A), \ldots, q_p(A)\}^\T$ is a $p$-dimensional vector of known  functions of covariates $A$, $\mu_{z,s_0s_1,a} \in \mathbb{R}^p$ is an unknown parameter vector, and ${\mathcal{L}}_{z,s_0s_1}^\ast(C)$ is an unspecified nonparametric function of $C$.

\end{assumption}

Assumption~\ref{assumption:NoInteraction} imposes additive separability between $A$ and $C$ in the model for ${\mathcal{L}}_{z,s_0s_1}(X)$, while  retaining modeling flexibility through the nonparametric component ${\mathcal{L}}_{z,s_0s_1}^\ast(C)$. This assumption generalizes several existing identification strategies, which emerge as special cases: 
\begin{itemize}
    \item[(i)] Instrumental variable: When $\mu_{0,\PP,a} = \mu_{1,\PP,a} = \mu_{1,\PN,a} = \mu_{0,\NP,a} = 0$, the covariate $A$ has no direct effect on the outcome $Y$, recovering the identification framework of \citet{Ding:2011}, \citet{Jiang2016}, and \citet{Wang2017iv}. 

    \item[(ii)] Linear model specification: When $q(A) = A$ and ${\mathcal{L}}_{z,s_0s_1}^\ast(C)$ is linear in $C$, it  reduces to the fully parametric linear model  specification studied by \citet{Ding:2011,Luo-multiarm-2021}.

\item[(iii)]  {No-interaction constraint:} \citet{Wang2017iv} assumed that, conditional on covariates $C$, neither $Z$ nor $G$ modifies the effect of covariates $A$ on  $Y$. In our setting, this would impose
$ {\mathcal{L}}_{1,01}(X) - {\mathcal{L}}_{1,00}(X) = {\mathcal{L}}_{1,11}(X) - {\mathcal{L}}_{1,10}(X) = {\mathcal{L}}_{0,11}(X) - {\mathcal{L}}_{0,10}(X)$ 
for all $X$, implying treatment effects are constant across strata. We do not impose this restriction. 
\end{itemize}

In practice, we recommend selecting covariates for $A$ based on two criteria: 
(i) choose variables that affect $Y$ primarily through strata $G$ rather than via direct causal pathways \citep{Ding:2011,Jiang2016,Wang2017iv}, analogous to the exclusion restriction in instrumental variables \citep{Angrist:1996}; 
(ii) select variables whose effects on $Y$ do not interact with $C$ and for which the functional form $q(A)$ can be reliably specified. 
Additionally, we empirically recommend that $q(A)$ employ flexible specifications such as polynomial terms, spline basis functions, or generalized linear model forms to capture as much information from $A$ as possible.

\begin{theorem}
\label{thm:identify}
Suppose Assumptions \ref{assumption:RCT}, \ref{assump:odds},  \ref{assumption:NoInteraction} hold, and for all real numbers $b_1, b_2, b_3, b_4$, the following conditions hold:
\begin{equation} \label{eq:iden-exp} \begin{aligned} & b_1 \cdot \pi_{\PP}(X) + b_2 \cdot \pi_{\PN}(X) + b_3 \cdot q(A) \pi_{\PP}(X) + b_4 \cdot q(A) \pi_{\PN}(X) = 0  \Rightarrow b_1 = b_2 = b_3 = b_4 = 0, \\ & b_1 \cdot \pi_{\PP}(X) + b_2 \cdot \pi_{\NP}(X) + b_3 \cdot q(A) \pi_{\PP}(X) + b_4 \cdot q(A) \pi_{\NP}(X) = 0   \Rightarrow b_1 = b_2 = b_3 = b_4 = 0. \end{aligned} \end{equation}
Then, (a) the ${\mathcal{L}}_{0,\PP}(X)$, ${\mathcal{L}}_{1,\PP}(X)$, ${\mathcal{L}}_{1,\PN}(X)$, and ${\mathcal{L}}_{0,\NP}(X)$ are identifiable; and (b) the Bayesian decision rule \eqref{eq:otr-bayesian} is identifiable.
\end{theorem}

Theorem \ref{thm:identify} essentially requires that the function sets  
$\{\pi_{\PP}(X), \pi_{\PN}(X), q(A)\pi_{\PP}(X), q(A)\pi_{\PN}(X)\}$ and $\{\pi_{\PP}(X), \pi_{\NP}(X), q(A)\pi_{\PP}(X), q(A)\pi_{\NP}(X)\}$ 
are linearly uncorrelated. A key advantage is that this linear independence condition \eqref{eq:iden-exp} can be empirically verified \citep{zhang2013assessing,ciocuanea2025sensitivity,tong2025semiparametric}, since $\pi_{\PP}(X)$, $\pi_{\PN}(X)$, and $\pi_{\NP}(X)$ are identifiable from observed data under Assumptions~\ref{assumption:RCT} and~\ref{assump:odds}. 
These conditions ensure the uniqueness of the solution and thus the identifiability of the optimal Bayesian decision rule.

Notably, beyond establishing identifiability, as discussed below  Assumption \ref{assump:odds}, Theorem~\ref{thm:identify} also provides a principled framework for sensitivity analysis. By specifying informative priors on the odds function $\theta(X)$ based on domain knowledge, researchers can systematically assess the robustness of their conclusions. We implement this framework by conducting sensitivity analysis under various specifications of $\theta(X)$ in both our simulation and real-data application. 

\section{Estimation and Inference}
\label{sec:estimation}

\subsection{Estimation Procedure}
\label{sec:est-proc}
We present the proposed estimation method for learning $G^*$. Let $e_0(X) = \Pr(S = 1 \mid Z = 0, X)$, and $e_1(X) = \Pr(S = 1 \mid Z = 1, X)$. The estimation procedures consist of the following five steps:
 \begin{itemize} 

  \item[{\it Step 1}.] Obtain the estimate \( \hat{\pi}_{s_0s_1}(X) \) for \( \pi_{s_0s_1}(X) \), where \( z, s_0, s_1 \in \{0, 1\} \). In practice, we consider two methods: 
         (i)  Given consistent estimators \( \hat{e}_0(X) \) and \( \hat{e}_1(X) \) for \( e_0(X) \) and \( e_1(X) \), compute \( \hat{\pi}_{s_0s_1}(X) \) using the equation~\eqref{eq:iden-pi11} in Supplementary Material;
        (ii) Use the EM algorithm for maximum likelihood estimation as proposed in \ref{EMalgorithm} of Supplementary Material. 
   
    \item[ {\it Step 2}.] Obtain the estimate \( \hat{{\mathcal{L}}}_{z,s_0s_1}(X) \) for \( {\mathcal{L}}_{z,s_0s_1}(X) \), where \( z, s_0, s_1 \in \{0, 1\} \). In particular, as indicated by equation \eqref{eq:residual-form} in Supplementary Material, 
    the ${\mathcal{L}}_{z,s_0s_1}(X)$ can be solved based on the following residual-based moment conditions: 
\[
\begin{aligned}
m_0(X; \beta_0) &= \mathbb{I}(Z=0,S=1)\, B (X) \left\{ Y - \sum_{s_1 \in \{0,1\}} \frac{\hat{\pi}_{1s_1}(X)}{\hat{e}_0(X)}\, {\mathcal{L}}_{0,1s_1}(X; \beta_{0,1s_1}) \right\}, \\
m_1(X; \beta_1) &= \mathbb{I}(Z=1,S=1)\, B (X) \left\{ Y - \sum_{s_0 \in \{0,1\}} \frac{\hat{\pi}_{s_01}(X)}{\hat{e}_1(X)}\, {\mathcal{L}}_{1,s_01}(X; \beta_{1,s_01}) \right\},
\end{aligned}
\]
where \( {\mathcal{L}}_{0,1s_1}(X; \beta_{0,1s_1}) \) and \( {\mathcal{L}}_{1,s_01}(X; \beta_{1,s_01}) \) denote working models for \( {\mathcal{L}}_{0,1s_1}(X) \) and \( {\mathcal{L}}_{1,s_01}(X) \), respectively, with unknown parameters \( \beta_{0,1s_1} \) and \( \beta_{1,s_01} \). The term \( B(X) \) is some user-specified functions whose dimension is at least as large as that of the parameter vectors.

   \item[{ \it Step 3.}]  
Given covariates $X=x$,  the treatment policy $\hat\rho( {\tilde{s}_0\tilde{s}_1},x)$ is determined based on the {policy} proposed in Section \ref{sec: Preliminaries}: 
\begin{equation*} 
    \hat\rho( {\tilde{s}_0\tilde{s}_1},x) = \mathbb{I}\left\{ \hat{{\mathcal{L}}}_{1,\tilde{s}_0\tilde{s}_1}(X) -   c_1(X) \geq \hat{{\mathcal{L}}}_{0,\tilde{s}_0\tilde{s}_1}(X) -  c_0(X)\right\}.
\end{equation*}
\item[{ \it Step 4.}] Applying the law of iterated expectations to \eqref{eq:misclassification-reward}, and using the estimated policy  $\hat{\rho}( {\tilde{s}_0\tilde{s}_1}, x)$ from { \it Step 3}, the estimated misclassification reward $\hat{\mathcal{R}}(\tilde{s}_0\tilde{s}_1 \mid s_0s_1)$ is calculated as:
\begin{align*} 
\begin{aligned} \hat{\mathcal{R}}&(\tilde{s}_0\tilde{s}_1\mid s_0s_1)
 = \mathbb{P}_n \left[ \begin{matrix}
          \{\hat{{\mathcal{L}}}_{1,s_0s_1}(X)- c_{1}(X)\}\hat\pi_{s_0s_1}(X)\hat\rho( {\tilde{s}_0\tilde{s}_1},X) \\+\{\hat{{\mathcal{L}}}_{0,s_0s_1}(X)- c_{0}(X)\}\hat\pi_{s_0s_1}(X)\{1-\hat\rho( {\tilde{s}_0\tilde{s}_1},X)\}  
     \end{matrix} \right] \Bigg/{\mathbb{P}_n\left\{\hat\pi_{s_0s_1}(X)\right\}},
\end{aligned}
\end{align*}   
where $\mathbb{P}_n(\cdot)$ denotes the sample average operator.

     \item[{ \it Step 5.}]  
Given the  observed covariates $X=x$, the estimated Bayesian decision rule, as stated in Theorem~\ref{thm:maxim-reward}, is given by
\begin{align*} 
\hat {G}^\ast   = \underset{  {s}_0 ,{s}_1\in\{0,1\}}{\mathrm{argmax}}\sum_{\tilde{s}_0 ,\tilde{s}_1\in\{0,1\} }{\hat{\mathcal{R}}}(\tilde{s}_0\tilde{s}_1\mid s_0s_1) \cdot \hat\pi_{\tilde{s}_0 \tilde{s}_1}(x).\end{align*}

 \end{itemize}

We give additional remarks on the implementation of {\it Step} 2. In practice, when parametric forms are assumed for \( {\mathcal{L}}_{z,s_0s_1}(X) \), the generalized method of moments \citep{Hansen:1982} provides a natural estimation framework. For nonparametric or semiparametric specifications, such as when ${\mathcal{L}}_{0,\PP}(X)$, ${\mathcal{L}}_{1,\PP}(X)$, ${\mathcal{L}}_{1,\PN}(X)$, and ${\mathcal{L}}_{0,\NP}(X)$ are modeled using sieve methods, we can estimate them using the sieve minimum distance method of \citet{AiChen2003} or the penalized sieve minimum distance method of \citet{ChenPouzo2015}. 

Beyond the estimation of pseudo-strata described above, we now discuss two important tasks that can be implemented using the estimated quantities from our estimation procedures:

  (i) \emph{Estimation of principal causal effects.} 
Given $\hat{\pi}_{s_0s_1}(X)$ from \textit{Step 1} and  $\hat{\mathcal{L}}_{z,s_0s_1}(X)$ from \textit{Step 2}, for $(z, s_0s_1) \in {(1,\PP), (1,\PN), (0,\PP), (0,\NP)}$, we estimate the marginal expected outcomes $\mathcal{L}_{z,s_0s_1}$  as follows:
\begin{equation}\label{ese:pse}
    \hat{\mathcal{L}}_{z,s_0s_1} =  {\mathbb{P}_n \left\{ \hat{\mathcal{L}}_{z,s_0s_1}(X) \cdot \hat{\pi}_{s_0s_1}(X) \right\}}\big/{\mathbb{P}_n\left\{\hat{\pi}_{s_0s_1}(X)\right\}}.
\end{equation}  
The principal causal effects within 
$G = 11$ is then given by $\hat{\mathcal{L}}_{1,11} - \hat{\mathcal{L}}_{0,11}$.
 {We will use \eqref{ese:pse} to evaluate the proposed estimation method for $L_{z, s_0s_1}(X)$.}

\textit{(ii) Personalized treatment policy.} 
Given $\hat{G}^\ast = {s}_0{s}_1$ obtained from \textit{Step 5}, the pseudo-strata-induced treatment policy is given by \begin{equation}
    \label{policy-rule}
    \hat{\rho}({s}_0{s}_1, x) = \mathbb{I}\left\{ \hat{\mathcal{L}}_{1,{s}_0{s}_1}(x) - c_1(x) \geq \hat{\mathcal{L}}_{0,{s}_0{s}_1}(x) - c_0(x)\right\},
\end{equation}
which assigns treatment if and only if the net profit under treatment exceeds that under control.

 \subsection{Inference}

We present the theoretical asymptotic properties of the proposed estimation method. First, we establish the convergence rate of the estimator $\hat{\mathcal{R}} (\tilde{s}_0\tilde{s}_1\mid s_0s_1)$, which forms the basis 
for analyzing the properties of the decision rule $\hat{G}^\ast$. The regularity conditions are summarized below.   
\begin{assumption}
\label{rate-cond} 
{Suppose   the following conditions hold:}
\begin{itemize}
    \item[(i)] The \( e_0(X), e_1(X), {\mathcal{L}}_{0,s_0 s_1}(X), {\mathcal{L}}_{1,s_0 s_1}(X) \) and their estimates $ \hat{e}_0(X), \hat{e}_1(X), \hat{{\mathcal{L}}}_{0,s_0 s_1}(X),$  $\hat{{\mathcal{L}}}_{1,s_0 s_1}(X)$ are uniformly bounded by a constant \( M > 0 \), i.e.,
    \[
    \begin{gathered}
        \|\hat{e}_{z}(X)\|_{\infty} \le M, ~ ~    \| {e}_{z}(X)\|_{\infty} \le M, ~~ 
        \|\hat{{\mathcal{L}}}_{z,s_0 s_1}(X)\|_{\infty} \le M, ~ ~  \|\hat{{\mathcal{L}}}_{z,s_0 s_1}(X) \|_{\infty} \le M,
    \end{gathered}
    \]
  {where $||\cdot||_\infty$ denotes the $L_\infty$ norm.}  
    
    \item[(ii)] For \( z \in \{0,1\} \) and \( s_0, s_1  \in \{0,1\} \), the estimated functions \( \hat{e}_z \) and \( \hat{{\mathcal{L}}}_{z,s_0 s_1} \) converge uniformly to their population counterparts at the following rates:
    \[
    \begin{aligned}
        \|\hat{e}_{z} (X)- e_{z}(X)\|_{\infty} &= O_p(n^{-\gamma_z}), ~
        \|\hat{{\mathcal{L}}}_{z,s_0 s_1} (X)- {\mathcal{L}}_{z,s_0 s_1}(X)\|_{\infty} &= O_p(n^{-\gamma_{z,s_0 s_1}}),
    \end{aligned} 
    \]
where all $\gamma_z > 0 $ and $  \gamma_{z,s_0 s_1} > 0.$
    \item[(iii)]  For each  \( s_0, s_1  \in \{0,1\} \), the density of the random variables  ${\mathcal{L}}_{1,s_0 s_1}(X) - {\mathcal{L}}_{0,s_0 s_1}(X)$  is uniformly bounded over $x \in \mathcal{X}$.

\end{itemize}
\end{assumption}

Assumptions \ref{rate-cond}(i) and (ii) present regularity conditions on the nuisance components involved in the estimation process, where \ref{rate-cond}(i) requires that these functions and their estimators be uniformly bounded by a positive constant, and \ref{rate-cond}(ii) specifies the uniform convergence rates of the estimators toward their population targets. 
Assumption \ref{rate-cond}(iii) introduces an additional, though mild, regularity condition requiring that the conditional expectation  ${\mathcal{L}}_{1,s_0 s_1}(X) - {\mathcal{L}}_{0,s_0 s_1}(X)$ have a uniformly bounded density. This assumption is a special case of the margin condition, which is widely used in the literature on estimating non-smooth functionals \citep{Audibert-etal-2007, Luedtke-etal-2016, 2018Who, Kennedy-2019}.  Further discussion on the convergence rate of the working models is provided in \ref{discussion-convengenc} of Supplementary Material. 

\begin{theorem}\label{thm:rate} 
Under Assumptions \ref{assumption:RCT}-\ref{rate-cond}, the estimated misclassification reward satisfies
\[
\left| \hat{\mathcal{R}}(\tilde{s}_0 \tilde{s}_1 \mid s_0 s_1) - \mathcal{R}(\tilde{s}_0 \tilde{s}_1 \mid s_0 s_1) \right| 
\leq O_p\left( n^{ - \gamma } \right),
\]where $\gamma = \min \left\{ \min_{z \in \{0,1\}} \gamma_z, \min_{(z, s_0 s_1)  } \gamma_{z,s_0 s_1}, 1/2\right\}$.
\end{theorem}

 Theorem \ref{thm:rate} shows that the estimation error of the misclassification reward is dominated by the slowest convergence rate among estimators of $\mathcal{L}_{0,s_0 s_1}(X)$ and $e_z(X)$ for $z, s_0, s_1 \in \{0, 1\}$. In particular, when we construct these estimators under correctly specified parametric models, $\hat{\mathcal{R}}(\tilde{s}_0 \tilde{s}_1 \mid s_0 s_1)$ achieves a convergence rate of $O_p(n^{-1/2})$.

Define  $\mathcal{Q}_{s_0s_1}(x) = \sum_{\tilde{s}_0, \tilde{s}_1 \in \{0,1\}} \mathcal{R}(\tilde{s}_0 \tilde{s}_1 \mid s_0s_1) \cdot \pi_{\tilde{s}_0 \tilde{s}_1}(x),$ and the true maximizer set as $
\mathcal{S}(x) = \arg\max_{s_0s_1} \mathcal{Q}_{s_0s_1}(x),$ 
which may contain more than one element (i.e., the maximum may not be unique). The following proposition provides the selection consistency.
\begin{proposition}
\label{prop:selection}
Under Assumptions \ref{assumption:RCT}-\ref{rate-cond},  and for any \( s_0s_1 \notin \mathcal{S}(x) \), there exists some \( \bar{s}_0\bar{s}_1 \in \mathcal{S}(x) \) such that $
\mathcal{Q}_{\bar{s}_0\bar{s}_1}(x) - \mathcal{Q}_{s_0s_1}(x) \ge \delta$ {for some} $ \delta > 0,$ 
we have
\[
\lim_{n \to \infty} \pr(\hat{G}^\ast \in \mathcal{S}(x) \mid X = x) = 1.
\]
\end{proposition}


Proposition \ref{prop:selection} establishes the selection consistency of the pseudo-strata estimator. The key condition imposes a strict separation requirement: every suboptimal candidate \( s_0s_1 \notin \mathcal{S}(x) \) must be uniformly dominated by at least one element \( \bar{s}_0\bar{s}_1 \in \mathcal{S}(x) \) in the optimal set, with a margin of at least $\delta > 0$. Under this  condition,  $\hat{G}^\ast$ asymptotically selects from the optimal decision set $\mathcal{S}(x)$ with probability one as the sample size grows, for any covariate configuration $X = x$.

       \section{Simulation Studies} 
    \label{sec:experiments} 
 \subsection{Simulation Settings}
 \label{sec:ER-sim} 

We conduct simulation studies to evaluate the finite-sample performance of the proposed method. The covariates $X=\{A,C\}$ 
are generated from a multivariate normal distribution 
$N( (0.25,-0.25)^\T, I_2)$, where $I_2$ is the $2\times2$ identity matrix. 
The treatment $Z$ is generated using a logistic model: 
\[
\Pr(Z=1 \mid  {X}) = \mathrm{expit}(0.15 - \delta A - \delta C),
\]
where $\delta \in \{0, 1\}$ controls the treatment assignment mechanism. When $\delta = 1$, treatment assignment depends on covariates $A$ and $C$; when $\delta = 0$, it is randomized. 
We specify a multinomial logistic model for the principal strata $G \in \{\NN, \NP, \PN, \PP\}$, with $G=\NN$ as the baseline stratification. The corresponding true strata probabilities are:
\begin{equation}
    \label{sensitivity:logistic} 
        \begin{gathered}
\Pr(G = \NN \mid X) \propto 1, ~ \Pr(G = \NP \mid X) \propto \exp(\iota_{\NP} + \iota_{\NP,A} A + \iota_{\NP,C} C), \\
\Pr(G = \PN \mid X) \propto \exp(\iota_{\PN} + \iota_{\PN,A} A + \iota_{\PN,C} C), ~
\Pr(G = \PP \mid X) \propto \exp(\iota_{\PP} + \iota_{\PP,C} C + \eta),
\end{gathered} 
\end{equation}
where $\iota_{\NP} = -0.3$, $\iota_{\NP,A} = -1$, $\iota_{\NP,C} = 0.5$, $\iota_{\PN} = 0.4$, $\iota_{\PN,A} = 1$, $\iota_{\PN,C} = -1$, $\iota_{\PP} = 0.1$, $\iota_{\PP,C} = -0.5$, and $\eta$ is the sensitivity parameter.
 The observed intermediate variable $S$  is then determined by the consistency assumption.   Additionally, by verifying the definition of the odds function, we have $\theta(X) = \exp(\eta)$; the verification details are provided in \ref{ssec:verification-thetaX} of Supplementary Material.

To mimic the non-negative nature of outcomes in the application (Section~\ref{sec:application}), the potential outcomes $Y_0$ and $Y_1$ are generated from non-negative distributions. Specifically, let $\varepsilon_0, \varepsilon_1  {\sim} U(-1, 1)$ independently. For $G = \NN$, $Y_0 = Y_1 = 0$; for $G = \NP$, $Y_0 = \exp(1.5 + A) + \exp(0.5 - 1.1C) + \varepsilon_0$ and $Y_1 = 0$; for $G = \PN$, $Y_0 = 0$ and $Y_1 = \exp(1 + A) + \exp(1.5 + 1.15C) + \varepsilon_1$; for $G = \PP$, $Y_0 = \exp(1.5 + A) + \exp(1 + 1.2C) + \varepsilon_0$ and $Y_1 = \exp(1 + A) + \exp(1 + 0.5C) + \varepsilon_1$.  The observed outcome $Y$ is determined by the consistency assumption. Assumptions \ref{assumption:RCT}--\ref{assumption:NoInteraction} hold in the above setups.}

We implement the proposed method described in Section~\ref{sec:est-proc}. In Step 1, we employ a multinomial logistic model for estimating $\pi_{s_0s_1}(X)$ with $s_0,s_1\in\{0,1\}$; implementation details are provided in \ref{sec:logistic-models} of Supplementary Material. In Step 2, we specify the outcome regression models ${\mathcal{L}}_{z,s_0s_1}(X)$ using an additive exponential form: 
\begin{equation}
    \label{eq:exp-models}
{\mathcal{L}}_{z,s_0s_1}(X) = \exp(\beta_{z,s_0s_1,0} + A) + \exp(\beta_{z,s_0s_1,1} + \beta_{z,s_0s_1,2} C),
\end{equation}
where the first exponential term captures the effect of covariate $A$, and the second term models the influence of covariate $C$, with all parameters varying across treatment levels and strata. 

In Supplementary Material, we also provide additional simulation studies under the same settings, except that potential outcomes $Y_0$ and $Y_1$ are generated using a   linear model. The results are similar to those in the main text, further confirming the robustness of the proposed method.

 \subsection{Simulation Results: Principal Causal Effects}
 
Table~\ref{tab: result-simulation-ER} reports the bias and standard error of the estimates of ${\mathcal{L}}_{z,s_0s_1}$ obtained using equation~\eqref{ese:pse} over 200 replications,  where $(z, s_0s_1) \in \{(1,\PP), (1,\PN), (0,\PP), (0,\NP)\}$.
From Table~\ref{tab: result-simulation-ER}, 
 the proposed method demonstrates robust and consistent performance under various combinations of parameters $(\delta, \eta)$. 
As the sample size increases from $n = 500$ to $n = 20000$, the bias converges to zero and the standard error decreases substantially, confirming the asymptotic properties of the estimators. 
When $\eta = 0$, corresponding to the conditional independence case $S_0 \indep S_1 \mid X$, the method exhibits smaller bias and SE even in small sample scenarios. When $\eta = 0.25$, corresponding to the conditional dependence case $S_0 \not\indep S_1 \mid X$ with positive correlation as discussed in Assumption~\ref{assump:odds}, 
The proposed method also performs well, although the standard error increases slightly due to the additional uncertainty introduced by the dependence structure. 

\begin{table}[t]
\caption{Simulation results for causal estimands ${\mathcal{L}}_{1,\AT}$, ${\mathcal{L}}_{1,\PN}$, ${\mathcal{L}}_{0,\AT}$, and ${\mathcal{L}}_{0,\NP}$, with bias ($\times 100$) and standard error ($\times 100$, in parentheses) reported.}
 \label{tab: result-simulation-ER} 
\centering
\setlength{\tabcolsep}{3pt}
\resizebox{0.9965878\columnwidth}{!}{%
\begin{tabular}{ccccccccccc}
\toprule
                        &  & \multicolumn{4}{c}{$\delta=0,\eta=0$}                     &  & \multicolumn{4}{c}{$\delta=0,\eta=0.25$}                  \\ \cline{3-6} \cline{8-11} 
Sample Size             &  & $n=$ 500     & $n=$ 2000    & $n=$ 5000    & $n=$ 20000   &  & $n=$ 500     & $n=$ 2000    & $n=$ 5000    & $n=$ 20000   \\
${\mathcal{L}}_{1,\AT}$ &  & -0.06 (1.06) & 0.01 (0.55)  & -0.02 (0.31) & -0.02 (0.15) &  & -0.07 (0.87) & 0.01 (0.43)  & -0.02 (0.24) & -0.01 (0.13) \\
${\mathcal{L}}_{1,\PN}$ &  & 0.07 (0.66)  & 0.01 (0.34)  & 0.02 (0.21)  & 0.01 (0.1)   &  & 0.09 (0.72)  & 0.02 (0.36)  & 0.02 (0.22)  & 0.01 (0.11)  \\
${\mathcal{L}}_{0,\AT}$ &  & 0.06 (0.42)  & 0.02 (0.23)  & 0.02 (0.14)  & 0.00 (0.06)     &  & 0.06 (0.36)  & 0.01 (0.20)   & 0.02 (0.12)  & 0.00 (0.06)     \\
${\mathcal{L}}_{0,\NP}$ &  & -0.01 (0.48) & -0.01 (0.23) & 0.00 (0.16)     & 0.01 (0.07)  &  & -0.04 (0.57) & 0.00 (0.28)     & 0.00 (0.17)     & 0.01 (0.08)  \\ \toprule
                        &  & \multicolumn{4}{c}{$\delta=0.25,\eta=0$}                  &  & \multicolumn{4}{c}{$\delta=0,\eta=0.25$}                  \\ \cline{3-6} \cline{8-11} 
Sample Size             &  & $n=$ 500     & $n=$ 2000    & $n=$ 5000    & $n=$ 20000   &  & $n=$ 500     & $n=$ 2000    & $n=$ 5000    & $n=$ 20000   \\
${\mathcal{L}}_{1,\AT}$ &  & -0.17 (1.13) & -0.04 (0.53) & -0.02 (0.33) & -0.01 (0.15) &  & -0.09 (0.88) & -0.04 (0.46) & -0.03 (0.25) & 0.00 (0.13)     \\
${\mathcal{L}}_{1,\PN}$ &  & 0.12 (0.71)  & 0.05 (0.32)  & 0.02 (0.22)  & 0.00 (0.10)      &  & 0.09 (0.73)  & 0.06 (0.37)  & 0.03 (0.22)  & 0.00 (0.11)     \\
${\mathcal{L}}_{0,\AT}$ &  & 0.04 (0.41)  & 0.02 (0.2)   & 0.01 (0.13)  & 0.00 (0.07)     &  & 0.06 (0.35)  & 0.01 (0.18)  & 0.02 (0.11)  & -0.01 (0.06) \\
${\mathcal{L}}_{0,\NP}$ &  & -0.02 (0.48) & 0.00 (0.22)     & 0.00 (0.16)     & 0.01 (0.07)  &  & -0.05 (0.52) & 0.00 (0.27)     & 0.00 (0.17)     & 0.02 (0.08)  \\ \toprule
\end{tabular}
}
\end{table}

\subsection{{Simulation Results:} Strata Classification Accuracy}
 \label{sec:policy-studies}
Since true principal strata $G$ are observable in simulation studies, we compare the proposed method with the posterior mode rule in terms of  strata classification accuracy. As described in Remark \ref{remark1}, the posterior mode rule assigns each individual to the principal stratum with the highest posterior probability: 
\begin{equation}
    \label{eq:post-ps}
    \hat{G}^\circ = \underset{ \tilde s_0, \tilde s_1 \in \{0,1\}}{\mathrm{argmax}}~ \hat\pi_{\tilde{s}_0 \tilde{s}_1}(X),
\end{equation}
where $\hat{\pi}_{\tilde{s}_0 \tilde{s}_1}(X)$ for $\tilde s_0, \tilde s_1 \in \{0, 1\}$ are estimated using the method in Section \ref{sec:est-proc}. Throughout all procedures, we set $c_0(X) = c_1(X) \equiv 0$. 

Figure~\ref{fig:case2-prop-plots} compares the strata classification accuracy between the proposed method (\texttt{proposed}, shown in blue) and the posterior mode approach (\texttt{posterior}, shown in red) under various combinations of parameters $(\eta, \delta)$, with a fixed sample size of $n = 20000$ over 200 replications. The classification accuracy is evaluated by computing $\Pr(\hat{G}^* = G)$ and $\Pr(\hat{G}^\circ = G)$. Across all scenarios, both methods achieve accuracy around 39\%--43\%. The \texttt{proposed} method consistently outperforms the \texttt{posterior} method: when $\eta = 0$, both methods achieve similar accuracy of approximately 42.5\%; however, when $\eta = 0.25$, the \texttt{proposed} method maintains accuracy around 40.5\%, while the \texttt{posterior} method drops to approximately 39\%, demonstrating the superior performance of the proposed approach in recovering true  strata.
\begin{figure}[t]
    \centering
    \includegraphics[width=0.9499\linewidth]{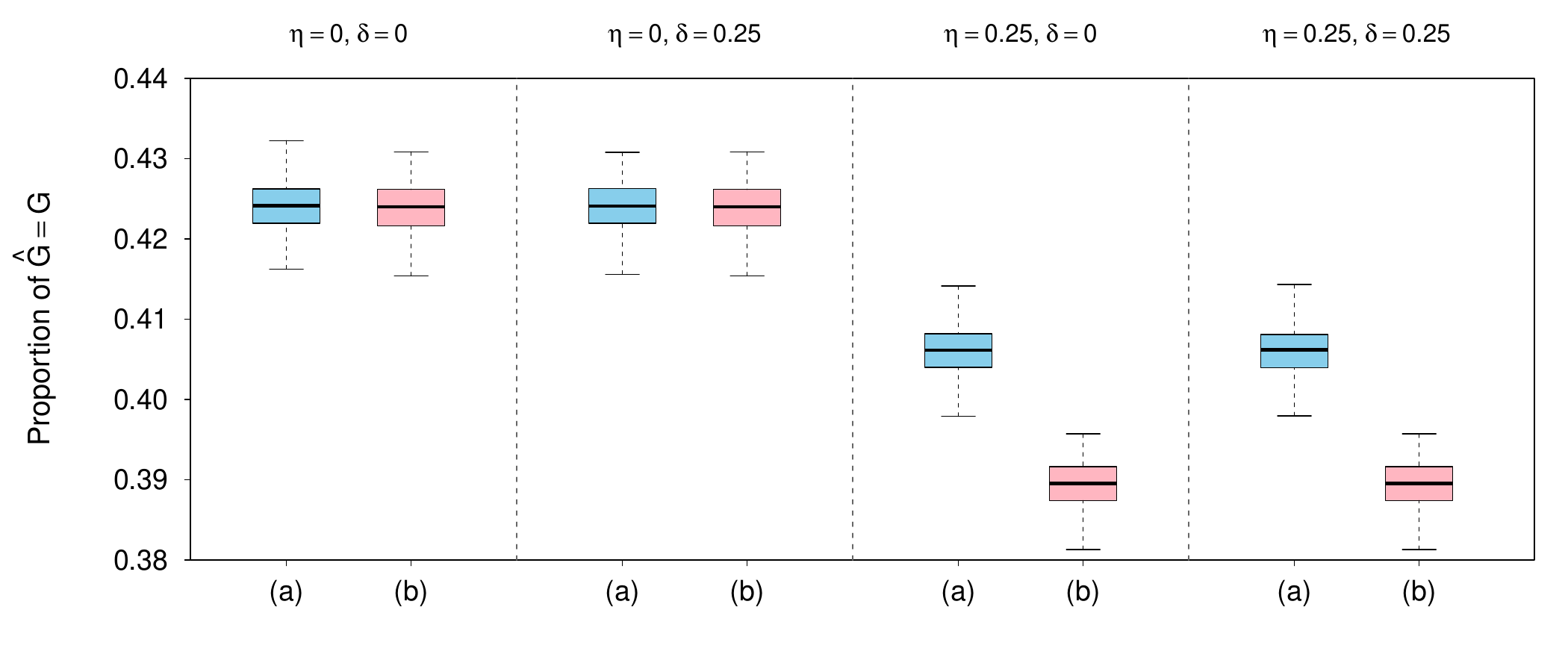}
   \caption{Comparison of principal strata classification accuracy under different parameters $(\eta, \delta)$. Blue boxes represent the proposed method, and red boxes represent the posterior mode method.}
    \label{fig:case2-prop-plots}
\end{figure}

 \begin{figure}[t]
    \centering
    \includegraphics[width=0.9499\linewidth]{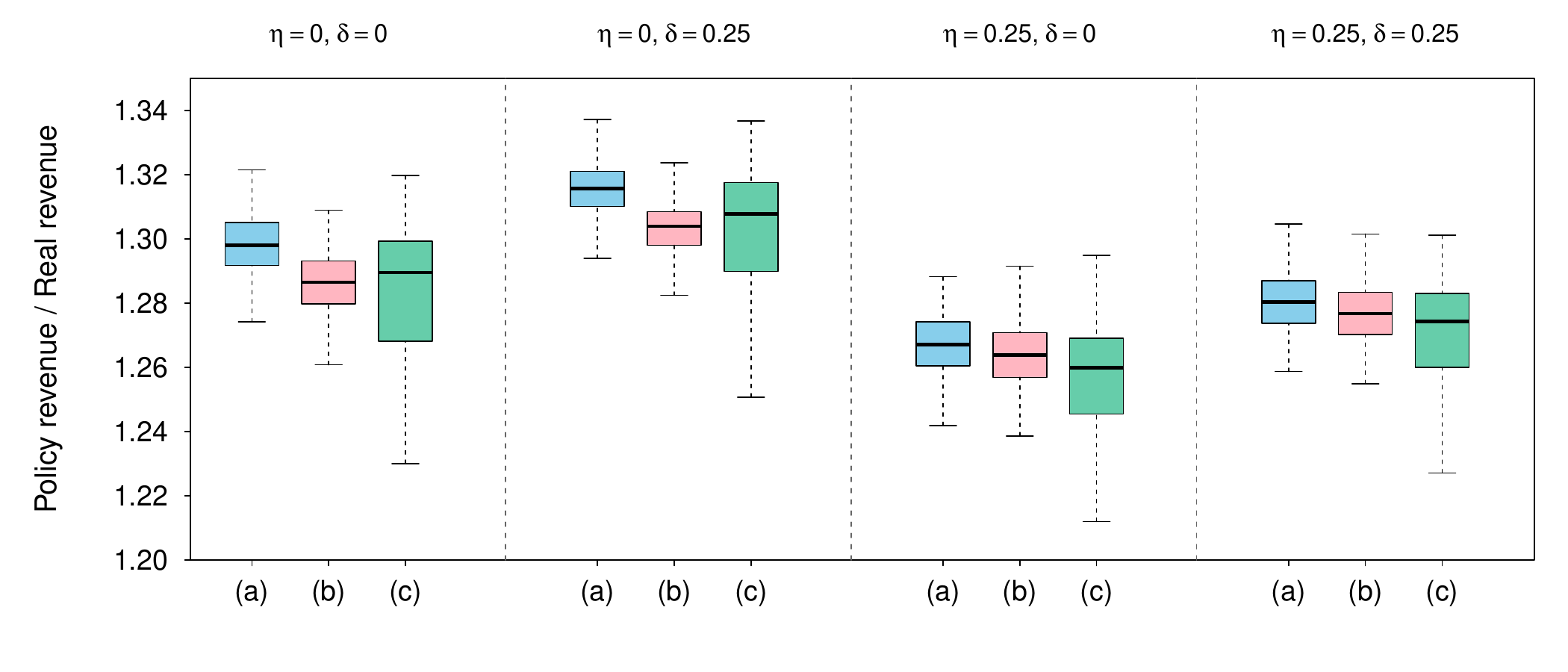}
\caption{Comparison of policy revenue under different parameter configurations ($\eta$, $\delta$). Red boxes represent the proposed method, blue boxes represent the posterior mode method, and green boxes represent the direct optimization method.}
    \label{fig:case2-ER-plots}
\end{figure}

\subsection{{Simulation Results:} Policy Revenue}

In this section, we compare three policy learning methods in terms of policy revenue.  
 The first two methods follow the same two-step process: 
  \begin{itemize}
      \item[\it Step 1:] estimating pseudo-strata using the approaches  described in Section~\ref{sec:policy-studies};
      \item[\it Step 2:]  assigning treatment according to the policy rule~\eqref{policy-rule}, which selects the treatment assignment  with higher expected net profit for each pseudo-strata.   
  \end{itemize}
The first two methods differ only in Step 1:  
the first one (\texttt{proposed}) uses the proposed pseudo-strata learning approach, while the second one  (\texttt{posterior}) uses the posterior-mode approach.


The third policy learning method (\texttt{direct}) directly optimizes the value function without estimating pseudo-strata. As introduced in Remark~\ref{remark2}, we substitute the estimated quantities $\hat{\pi}_{s_0s_1}(X)$ and $\hat{\mathcal{L}}_{z,s_0s_1}(X)$ into the value function~\eqref{eq:policy-value} to obtain the empirical value $\hat{V}(d)$. We then search for the optimal policy within a linear policy class: $
\hat{d}^*(X) = {\arg\max}_{d \in \Pi}~ \hat{V}(d),$
where $\Pi = \{d(X) = \mathbb{I}\{X^\T \beta > 0\} : \beta \in \mathbb{R}^p\}$ represents linear decision rules parameterized by  $\beta$. Treatment is assigned when the linear score $X^\T\beta$ is positive. We employ grid search over a discretized parameter space to identify the optimal $\beta^*$ that maximizes $\hat{V}(d)$.

Figure~\ref{fig:case2-ER-plots} compares policy revenue across the three methods under various parameter configurations $(\eta, \delta)$ with sample size $n = 20000$ over 200 replications. The figure reports the ratio of policy revenue to actual revenue, where policy revenue denotes the expected outcomes under each method's treatment policy, and actual revenue corresponds to the mean of observed outcomes $Y$. 
From Figure~\ref{fig:case2-ER-plots}, the \texttt{proposed} achieves the highest median revenue (1.28--1.32) across all settings, outperforming both \texttt{posterior} (1.26--1.30) and \texttt{direct} (1.25--1.31). 
 The superior performance of \texttt{proposed} relative to \texttt{direct} may stem from its higher stratum classification accuracy.

 \section{Application}
 \label{sec:application}
 \subsection{Background}

Since 2023, Criteo has released anonymized click and conversion logs from its Predictive Search platform to support research on conversion modeling in sponsored search advertising \citep{tallis2018reacting,jeunen2022consequences}. The dataset covers 90 days of real user traffic, amounting to approximately 6.4 GB, and records user clicks during advertising campaigns, along with whether those clicks eventually resulted in a purchase within a 30-day window. Each record captures a user's click on a product advertisement shown after the user expressed purchase intent through an online search engine. The dataset includes product features (such as target age group, brand, gender, and price), click time (uniformly time-shifted), user attributes, and device information. It also contains an indicator of whether a conversion occurred and the time interval between the click and the conversion.

In this study, we excluded all records containing missing values and randomly sampled 20000 complete entries for modeling and analysis. We selected nine categorical variables closely related to user-product interaction and conversion behavior, denoted as $X_1$ through $X_9$. Specifically, $X_1$ indicates the target age group for the product (e.g., adult, teen, or child), $X_2$ represents the device  type  used by the user at the time of the click (e.g., desktop, mobile, or tablet),  $X_3$ corresponds to the product’s intended gender (e.g., male, female, or unisex), $X_4$ represents the product brand,  $X_5$ through $X_8$ denote four hierarchical product category levels, ranging from broad to specific, and $X_9$ indicates the country in which the advertisement was displayed or the product was sold. According to the official documentation, all categorical variables have been hashed to protect the privacy of users and advertisers, rendering their exact meanings unidentifiable. Thus, we encoded these variables as integers to facilitate statistical modeling.

We use three key variables for subsequent analysis. 
 The first variable is the binary treatment indicator $Z$, which captures user engagement intensity. Specifically, we define $Z = 1$ if the number of product clicks in the past week is less than 20, i.e., $Z = \mathbb{I}(\texttt{nb clicks 1week} < 20)$. The threshold of 20 is approximately the median in the dataset and helps distinguish low-frequency from high-frequency users. When the number of clicks is below 20 ($Z = 1$), it may reflect limited interest or low ad exposure. In contrast, $Z = 0$ indicates high-frequency engagement, which, considering that ad exposures are typically triggered by search intent, may suggest intensive ad delivery, strong product interest, or both. 
The second variable $S$ indicates whether a conversion occurred, that is, whether the user made a purchase. 
The third variable $Y$ is a continuous outcome defined as product revenue. 

 \subsection{{Estimation Results:} Principal Causal Effects}\label{app:Point estimation} 
 \begin{figure}[t]
    \centering
        \includegraphics[width=0.89590495\linewidth]{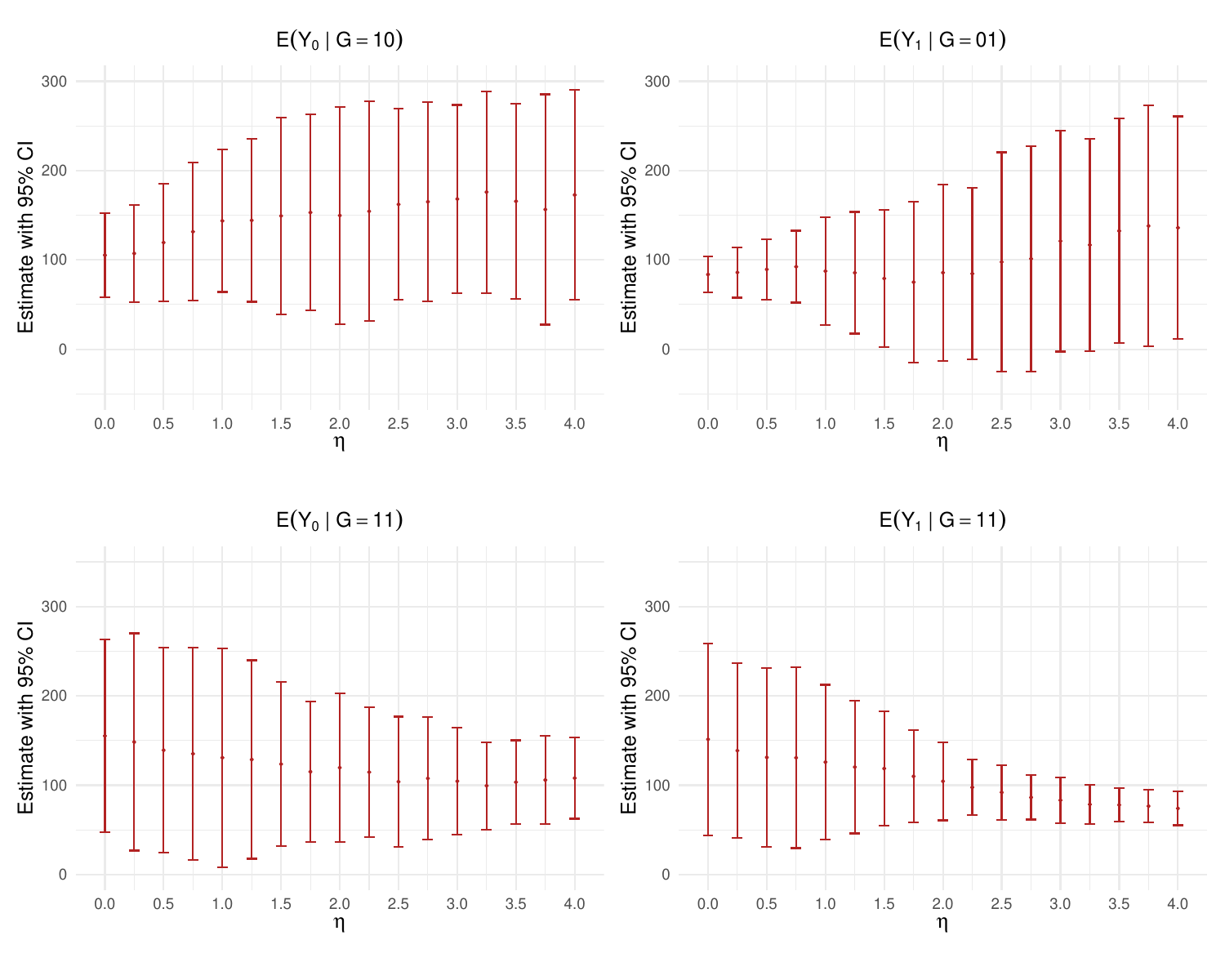}
\caption{Sensitivity analysis for principal causal effects, where $\eta$ is the sensitivity parameter in \eqref{sensitivity:logistic}. } 
    \label{fig:enter-label-realdata}
\end{figure}

We examine whether user engagement intensity, measured by how actively users click on ads, influences the platform's overall revenue. We apply the estimation procedure described in the simulation studies, estimating the principal causal effect quantities ${\mathcal{L}}_{z,s_0s_1}$ (marginal expected outcomes) via the estimators $\hat{\mathcal{L}}_{z,s_0s_1}$ in \eqref{ese:pse} for $(z, s_0s_1) \in {(1,\PP), (1,\PN), (0,\PP), (0,\NP)}$. 
Following the simulation setup in \eqref{sensitivity:logistic}, we specify the odds ratio function in Assumption~\ref{assump:odds} as $\theta(X) = \exp(\eta)$. As outlined in Assumption~\ref{assump:odds}, we vary $\eta$ as a sensitivity parameter to examine how the degree of conditional dependence between $S_0$ and $S_1$ affects the conclusions.

We anticipate a positive association between the two potential purchase statuses. Users who would purchase under low  engagement $(S_1 = 1)$ likely possess strong intrinsic purchase intent and are also more responsive to intensive advertising $(S_0 = 1)$, as the additional exposure serves to reinforce their existing interest and further increase conversion probability. Conversely, users who would not purchase under low engagement ($S_1=0$) typically lack intrinsic motivation and are also unlikely to be persuaded to purchase by intensive ad exposure ($S_0=1$). Under this positive association, individuals with high intrinsic purchase intent exhibit greater responsiveness to intensive advertising compared to those with low intrinsic intent. We therefore examine the sensitivity parameter $\eta$ ranging from 0 (conditional independence) to 3.5 (strong positive association).

 \begin{figure}[t]
    \centering
    \includegraphics[width=0.8\linewidth]{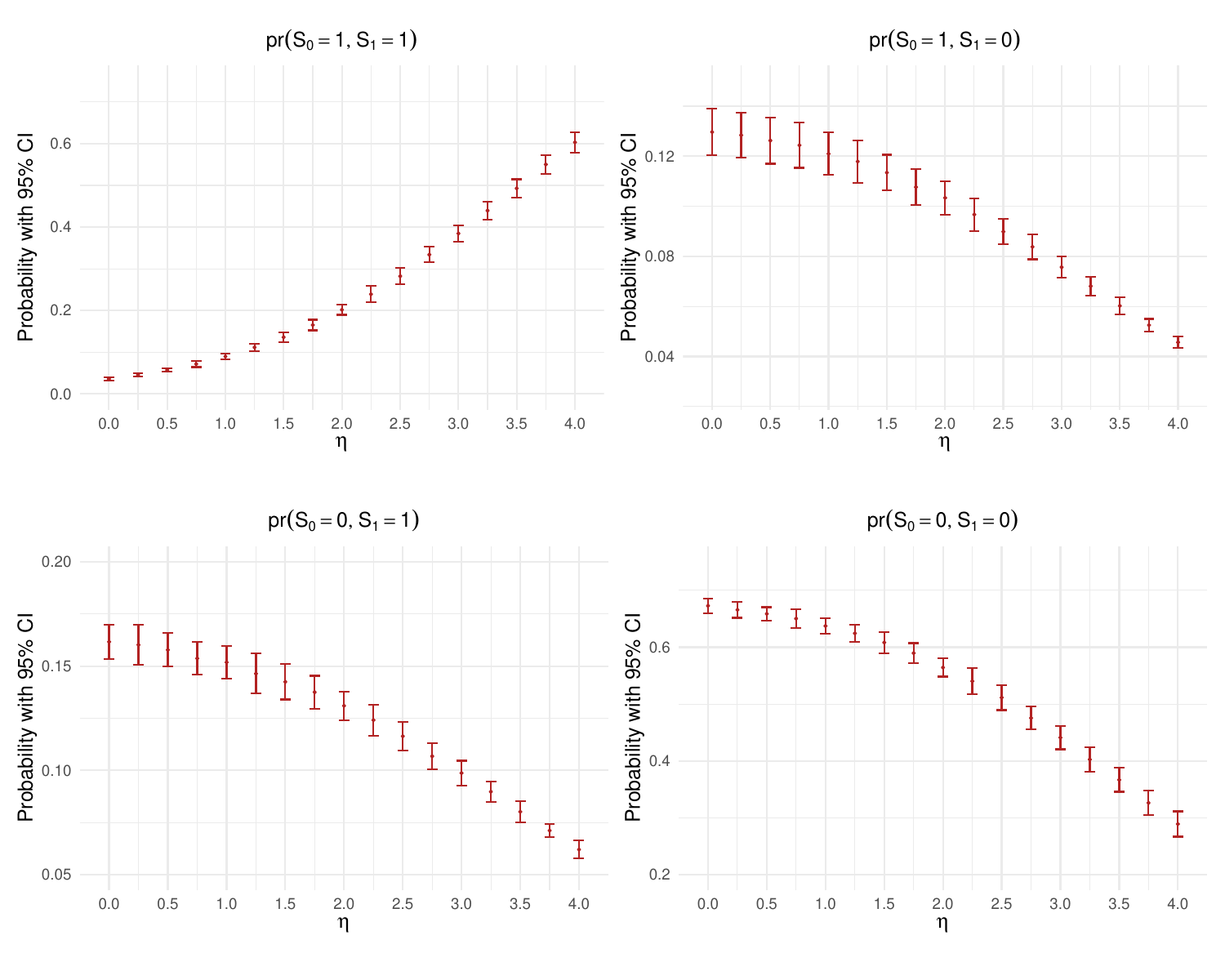}
\caption{Sensitivity analysis of the proportions of the principal strata.}
    \label{fig:enter-label-prob}
\end{figure}
 
To model the conditional outcome means ${\mathcal{L}}_{z,s_0s_1}(X)$, we treat the product's target gender ($X_3$) as the binary variable $A$ in Assumption~\ref{assumption:NoInteraction}, with the remaining covariates as $C$. We assume that gender's effect on revenue $Y$ does not interact with other covariates $C$, the effects of device   type, brand, or age on revenue remain consistent across gender groups. This no-interaction assumption is plausible in our setting and ensures Assumption~\ref{assumption:NoInteraction} is satisfied. We therefore adopt the separable model in equation~\eqref{eq:exp-models} for estimation.

Figure~\ref{fig:enter-label-realdata} presents the estimates of ${\mathcal{L}}_{z,s_0s_1}$ under varying values of the sensitivity parameter $\eta$ in \eqref{sensitivity:logistic} over 200 bootstrap resamples. Additional numerical results are provided in Section~\ref{sec:addition-results-application} of Supplementary Material. Across all panels, the point estimates of ${\mathcal{L}}_{z,s_0s_1}$ remain relatively stable as $\eta$ varies, and the corresponding 95\% confidence intervals almost always exclude zero, suggesting that expected revenue remains robustly positive.  However, the behavior of confidence interval width varies across different strata: for ${\mathcal{L}}_{1,01}$ and ${\mathcal{L}}_{1,11}$, the intervals widen substantially as $\eta$ increases, reflecting greater uncertainty as conditional dependence strengthens; in contrast, for ${\mathcal{L}}_{0,10}$ and ${\mathcal{L}}_{0,11}$, the confidence intervals remain relatively stable across different $\eta$. 



These results reveal the expected outcomes for each stratum under their observable treatment assignments. For always buyers ($G = 11$, lower panels), expected revenue ranges from approximately 50--150 under control and 100--250 under treatment, showing consistently positive outcomes regardless of ad exposure. The persuadable subgroup ($G = 01$, upper-right) shows expected revenue of 100--200 under treatment. For the discouraged subgroup ($G = 10$, upper-left), expected revenue under control ranges from 50--150. Across all strata, as $\eta$ increases, confidence intervals for always buyers narrow substantially, while those for persuadables and discouraged buyers widen considerably. This pattern directly reflects the estimated strata  proportions discussed in the next paragraph: always-buyers increase from 5\% to over 60\%, while persuadables and discouraged buyers remain small (below 16\%) and gradually decline.

In Figure \ref{fig:enter-label-prob}, we present sensitivity analysis for the estimated strata proportions $\pi_{s_0s_1}$ under varying $\eta$  over 200 bootstrap resamples. 
 As positive dependence strengthens ($\eta$ increases from 0 to 4), the proportion of always buyers $\pi_{11}$ rises from 5\% to over 60\%, while never buyers $\pi_{00}$ decline from 65\% to 30\%. The persuadable $\pi_{01}$ and discouraged $\pi_{10}$ subgroups remain small (5--15\%) but show opposite trends: persuadables slightly increase while discouraged decrease.

\subsection{{Estimation Results:} Policy Revenue} 
 \begin{figure}[t]
    \centering
    \includegraphics[width=0.995\linewidth]{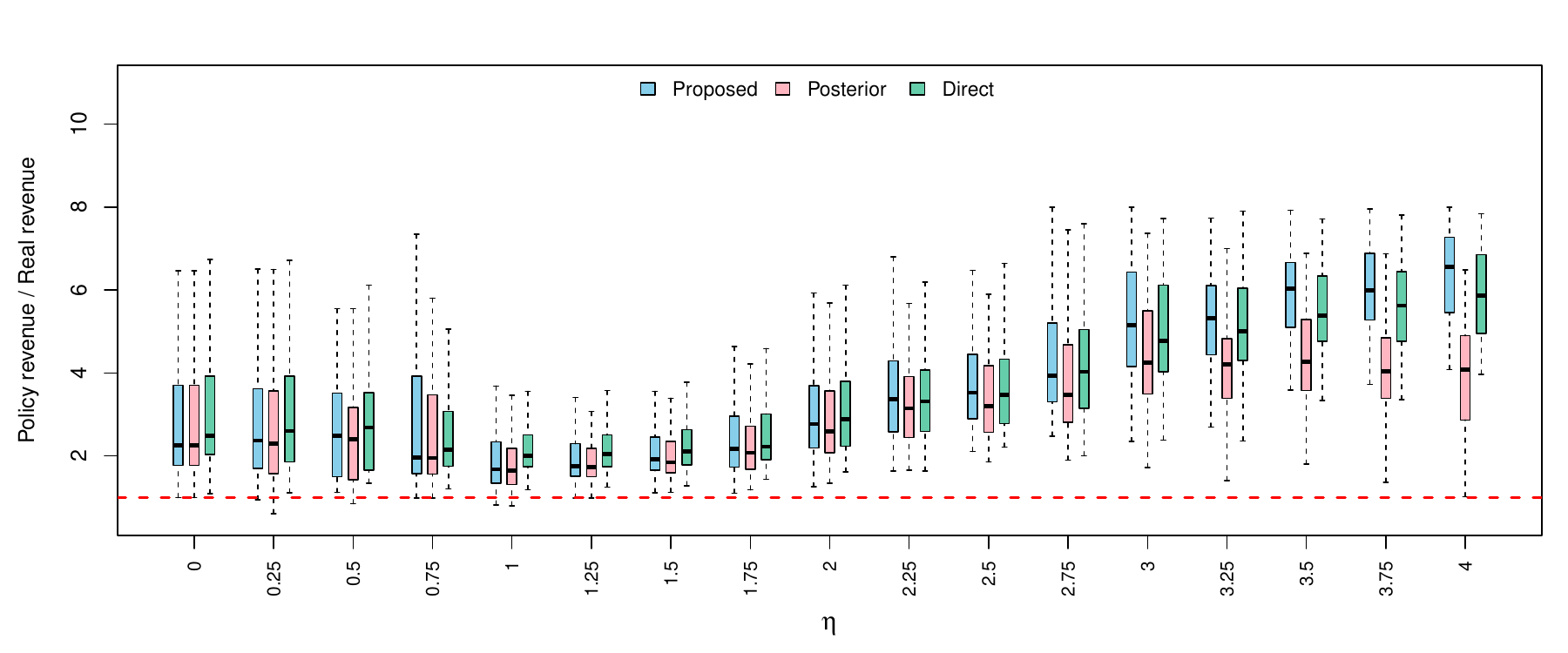}
\caption{Sensitivity analysis of revenue ratios for three methods. Blue: proposed method; Red: posterior mode method; Green: direct optimization.}
    \label{fig:enter-label-ratio}
\end{figure}

Following the approach in Section~\ref{sec:policy-studies}, we evaluate the learned pseudo-strata using policy revenue. Figure~\ref{fig:enter-label-ratio} compares the ratio of policy revenue to actual revenue across three methods over 200 bootstrap resamples: the proposed pseudo-strata learning method (blue boxes), the posterior mode rule (red boxes), and the direct policy optimization approach (green boxes), under different values of the sensitivity parameter $\eta$ in \eqref{sensitivity:logistic}. The horizontal dashed red line at 1 represents the benchmark where policy revenue equals actual revenue.

When $\eta <1$, all three methods perform comparably, with policy revenue ratios between 1.5 and 4. However, the results show substantial instability due to challenges in estimating principal strata proportions. As shown in Figure \ref{fig:enter-label-prob}, never buyers comprise approximately 65\% of the population when $\eta <1$, while always buyers, discouraged buyers, and persuadable buyers together account for only about 30\%. This small proportion of responsive individuals leads to high estimation variance. 
As $\eta$ increases from 0 to  2.25, the three methods begin to differ. At $\eta = 2.25$, the proposed method shows clear advantages, coinciding with a more balanced population structure: always buyers increase to roughly 25\% while never buyers decrease to about 55\%, leading to more stable estimation overall. 
For $\eta >3$, the proposed method outperforms the other two methods.  However, the results again become unstable in this region, likely because discouraged buyers and persuadable buyers together account for only about 10\% of the population. Nonetheless, the proposed method maintains competitive performance across all scenarios. This aligns with our simulation findings: the proposed method's explicit modeling of the latent strata enables better policy recommendations, while the direct optimization and posterior mode approaches cannot fully exploit this  structure.



\section{Discussion}
\label{sec: disc}
In this paper, we propose a novel method for inferring individuals' strata, defined by the joint value of the potential outcomes of a post-treatment response variable. 
The core idea is to incorporate an outcome variable observed after the principal strata to guide the learning process, while also accounting for the {strata} probabilities within a Bayesian decision framework.   
While the learned pseudo-strata may not perfectly match the true latent strata, our method incorporates misclassification rewards into the learning objective, ensuring that even imperfect pseudo-strata lead to effective treatment policies and substantial revenue gains in practice.

Future research directions include relaxing the current identification assumptions and adopting more flexible model structures to enhance the method's applicability across a wider range of scenarios.  The exploration of these issues is beyond the scope of this paper, and we consider them important avenues for future research.

\section*{Acknowledgement}
 Peng Wu is the corresponding author. The authors thank Prof. Qingyuan Zhao and Dr. Yue Zhang for their valuable discussions.

\section*{Supplementary Material}
\label{SM}
The Supplementary Material includes all technical proofs,  {along with additional results for the simulation and  application.} 

\bibliographystyle{chicago}
\bibliography{paper-ref}

\vspace{0.5cm}
\noindent

 \newpage

\renewcommand{\theproposition}{S\arabic{proposition}}
\renewcommand{\thetheorem}{S\arabic{theorem}}
\renewcommand{\theassumption}{S\arabic{assumption}}
\renewcommand{\thesection}{S\arabic{section}}
\renewcommand{\theequation}{S\arabic{equation}}
\renewcommand{\thelemma}{S\arabic{lemma}} 
\renewcommand{\thetable}{S\arabic{table}}  
\renewcommand{\thefigure}{S\arabic{figure}}

\renewcommand{\theHtheorem}{S\arabic{theorem}}
\renewcommand{\theHsection}{S\arabic{section}}
 
\setcounter{section}{0} 
 \begin{center}
     \huge{\bf Supplementary Material}
 \end{center}

The Supplementary Material includes all technical proofs,  {along with additional results for the simulation and  application.} 

\section{The proof of Theorem \ref{thm:maxim-reward}}

 \begin{proof}
   For any pseudo-stratum $\tilde G  = \tilde{s}_0\tilde{s}_1$, there exists a corresponding region ${\tilde{D}}_{\tilde{s}_0\tilde{s}_1}$ in the covariate space.  
Let ${\tilde{D}} = \{{\tilde{D}}_{00}, {\tilde{D}}_{01}, {\tilde{D}}_{10}, {\tilde{D}}_{\PP}\}$ be the partition defined by the  pseudo-stratum $\tilde{G}$.

When the true stratum of an individual is $G = s_0s_1$, but the decision rule classifies the individual into $\tilde{G} = \tilde{s}_0\tilde{s}_1$ (i.e., the baseline covariates $X $ fall into region ${\tilde{D}} _{\tilde{s}_0\tilde{s}_1}$, with $\tilde{G} \ne G$), then a misclassification occurs. 
We then consider the misclassification probability of assigning a sample truly belonging to true stratum $G = s_0s_1$ to the pseudo-stratum $\tilde{G} = \tilde{s}_0\tilde{s}_1$ as $\Pr(\tilde{G} = \tilde{s}_0\tilde{s}_1 \mid G = s_0s_1 )$. Clearly, we have:
$$
\Pr(\tilde{G} = \tilde{s}_0\tilde{s}_1 \mid G = s_0s_1 ) = \int_{{\tilde{D}}_{\tilde{s}_0\tilde{s}_1}} f (X\mid G={s_0s_1}) \, dX,
$$where $f(X \mid G = s_0s_1)$ denotes the conditional density of the covariates $X$ given the true stratum  $G = s_0s_1$.
From  \eqref{eq:decision-rule} in the main text, we have: 
\begin{align*}
 V(\tilde{G} ) &= \sum_{  {s}_0 ,{s}_1\in\{0,1\} } \Pr(G =s_0s_1)\sum_{\tilde{s}_0 ,\tilde{s}_1\in\{0,1\} } \Pr(\tilde{G} = \tilde s_0 \tilde s_1 \mid G = s_0s_1)\mathcal{R}(\tilde{G} = \tilde s_0 \tilde s_1\mid G={s_0s_1})\\
 & = \sum_{  {s}_0 ,{s}_1\in\{0,1\} } \pr(G={s_0s_1}) \sum_{\tilde{s}_0 ,\tilde{s}_1\in\{0,1\} } \mathcal{R}(\tilde{G} = \tilde s_0 \tilde s_1\mid G={s_0s_1}) \int_{{\tilde{D}}_{\tilde s_0 \tilde s_1}} f(X\mid G={s_0s_1}) dX\\&
= \sum_{\tilde{s}_0 ,\tilde{s}_1\in\{0,1\} } \int_{{\tilde{D}}_{\tilde s_0 \tilde s_1}} \sum_{  {s}_0 ,{s}_1\in\{0,1\} } \pr(G={s_0s_1}) \mathcal{R}(\tilde{G} = \tilde s_0 \tilde s_1\mid G={s_0s_1}) f(X\mid G={s_0s_1}) dX\\&
= \sum_{\tilde{s}_0 ,\tilde{s}_1\in\{0,1\} } \int_{{\tilde{D}}_{\tilde s_0 \tilde s_1}} r_{\tilde s_0 \tilde s_1 }(X) dX,
\end{align*}
where\begin{align*}
    r_{\tilde s_0 \tilde s_1}(X)&=\sum_{  {s}_0 ,{s}_1\in\{0,1\} } \pr(G={s_0s_1}) \mathcal{R}(\tilde{G} = \tilde s_0 \tilde s_1\mid G={s_0s_1}) f(X\mid G={s_0s_1})  
    \\&=\sum_{  {s}_0 ,{s}_1\in\{0,1\} } \pr(G={s_0s_1}\mid X) \mathcal{R}(\tilde{G} = \tilde s_0 \tilde s_1\mid G={s_0s_1}) f(X)  
    \\&= h_{\tilde s_0 \tilde s_1}(X)f(X) ,
\end{align*}
and $$
h_{\tilde s_0\tilde s_1}(X) = \textstyle\sum_{ {s}_0, {s}_1\in\{0,1\}} 
\mathcal{R}(\tilde{s}_0\tilde{s}_1 \mid s_0s_1) \cdot 
\pi_{ {s}_0 {s}_1}(X).$$ 
Therefore, for any pseudo-stratum $\tilde{G}$ and its  partition ${\tilde{D}}= \{{\tilde{D}}_{00}, \dots,{\tilde{D}}_{\PP}\}$ of $\mathbb{R}^m$, the average misclassfication reward is:
\begin{align*}
  V(\tilde{G})= \sum_{\tilde{s}_0 ,\tilde{s}_1\in\{0,1\} } \int_{\tilde{D}_{\tilde s_0 \tilde s_1}} r_{\tilde s_0 \tilde s_1 }(X) dX= \sum_{\tilde{s}_0 ,\tilde{s}_1\in\{0,1\} } \int_{{\tilde{D}}_{\tilde s_0 \tilde s_1}} h_{\tilde s_0 \tilde s_1}(X)f(X)dX.
\end{align*}
Recall that the optimal decision region $D^\ast_{s_0s_1}$ is defined as   $$D_{s_0s_1} ^\ast= \left\{ X : h_{s_0s_1}(X) \geq h_{\bar{s}_0\bar{s}_1}(X) \quad \text{for all} \quad  \bar{s}_0, \bar{s}_1\in\{0,1\}  \right\} .$$ Thus, for any $\tilde{G}$, we have
\begin{align*}
    V(G^\ast ) -  V(\tilde{G})&= \sum_{ s_0, s_1\in\{0,1\} } \int_{D^\ast_{ s_0  s_1}} r_{ s_0  s_1 }(X) dX -   \sum_{\tilde s_0,\tilde s_1\in\{0,1\} } \int_{{\tilde{D}}_{\tilde s_0 \tilde s_1}} r_{\tilde s_0 \tilde s_1 }(X) dX\\
&=\sum_{ s_0, s_1\in\{0,1\} }  \sum_{\tilde s_0,\tilde s_1\in\{0,1\} }\int_{D^\ast_{ s_0  s_1}  \cap {\tilde{D}}_{\tilde s_0 \tilde s_1}} \{r_{ s_0  s_1}(X) - r_{\tilde s_0 \tilde s_1 }(X)\} dX\\
&=\sum_{ s_0, s_1\in\{0,1\} }  \sum_{\tilde s_0,\tilde s_1\in\{0,1\} }\int_{D^\ast_{ s_0  s_1}  \cap {\tilde{D}}_{\tilde s_0 \tilde s_1}} \{h_{ s_0  s_1}(X) - h_{\tilde s_0 \tilde s_1 }(X)\} f(X)dX.
\end{align*}
From the definition of $D_{  s_0  s_1}^\ast $, we have $h_{  s_0   s_1} (X) \geq h_{\tilde s_0 \tilde s_1 }(X)$ for all $X \in D^\ast_{  s_0 s_1} $. Hence,
\[
 V(G^\ast ) -  V(\tilde{G})\geq 0,
\]
which implies $
 V(G^\ast )  = \max_{\tilde{G}} V(\tilde{G}).$ 
Thus, $G^\ast$ is the optimal Bayesian decision rule.
 \end{proof}
\section{The proof of Theorem  \ref{thm:identify}}
\subsection{The identification of proportions of true strata}    The identification of the proportions of true strata has been established in several prior works \citep{zhang2013assessing,ciocuanea2025sensitivity,tong2025semiparametric,Wu-etal-2025-Harm}. For completeness, we provide a proof in Lemma \ref{lem:prop-ps}.
\begin{lemma}\label{lem:prop-ps}
    Under Assumptions \ref{assumption:RCT} and \ref{assump:odds}, the proportions of true strata are identifiable.
\end{lemma}
\begin{proof}
Under the treatment ignorability assumption (Assumption~\ref{assumption:RCT}), for any given odds function $\theta(X)$ in Assumption \ref{assump:odds}, the proportion of individuals in the always buyer stratum $G=11$ can be identified using Equation (6) in \citet{lipsitz1991generalized} as follows: 
\begin{equation}
\label{eq:iden-pi11}
\pi_{\PP}(X) =\left\{\begin{array}{ll}
\dfrac{ \theta (X) + \{1 - \theta (X) \}\{e_0(X) + e_1(X)\} -\theta (X)  \sqrt{\delta(X)} }
{ 1 - \theta (X)  },& \theta (X) \neq 1,\\\addlinespace[1mm]
e_0(X)  e_1(X),& \theta (X) = 1.
\end{array}\right. 
\end{equation}
where $ 
\delta(X) = \left[1 + \{\theta^{-1}(X) - 1\}\{e_0(X) + e_1(X)\}\right]^2 
- 4\theta(X)\{\theta^{-1}(X) - 1\}e_0(X)e_1(X).$ 

Once $\pi_{\PP}(X)$ is identified in \eqref{eq:iden-pi11}, the remaining proportions can be obtained as:
\begin{equation*}
    \begin{gathered}
    \pi_{\CO}(X)=e_1(X)-\pi_\AT(X),~ \pi_{\DE}(X)=e_0(X)-\pi_\AT(X),\\ \pi_{\NT}(X)=1-\pi_\AT(X)-\pi_\CO(X)-\pi_{\DE}(X).
\end{gathered}
\end{equation*}
\end{proof}
\subsection{The identification of conditional outcome expectations} \begin{lemma}\label{lem:prop-pce}
   Under Assumptions \ref{assumption:RCT}, {\red \ref{assumption:NoInteraction}}, \ref{assump:odds}, and condition \eqref{eq:iden-exp},  the conditional outcome mean ${\mathcal{L}}_{z,s_0s_1}(X)$ is identifiable for $(z,s_0s_1)\in\{(1,\PP),(1,\PN),(0,\PP),(0,\NP)\}$.
\end{lemma}
\begin{proof} We present a unified proof for the identification of the conditional outcome expectations, as stated in Theorem \ref{thm:identify}. For any covariate vector $X=(A,C)$, define
\[
\begin{aligned}
    \omega_1(x) = \pr(G = \PP \mid Z = 1, S = 1, x) = {\pi_{\PP}(x)}/{e_1(x)},\\
    \omega_0(x) = \pr(G = \PP \mid Z = 0, S = 1, x) = {\pi_{\PP}(x)}/{e_0(x)}.
\end{aligned}
\]
The numerator can be identified based on Lemma \ref{lem:prop-ps}, and the denominator is identifiable from the observed data under Assumptions~\ref{assumption:RCT}.  Let $x_j=(c,a_j)$ for $j=1,\ldots,2\cdot  p+2$, we thus have
\begin{equation}\label{eq:residual-form}
  \begin{aligned}
        \E (Y  &\mid Z=S=1,x_1) \\&=   \E (Y   \mid G=11,Z=S=1,x_1) \omega_1(x_1)  + \E (Y   \mid G=01,Z=S=1,x_1)  \{1- \omega_1(x_1) \}\\&=      {\mathcal{L}}_{{{1,\PP}}}(x_1)  \omega_1(x_1) +  {\mathcal{L}}_{{1,\PN}}(x_1) \{1- \omega_1(x_1) \}
        \\&=     \mu_{1,\PP,a}^\T q(a_1)\omega_1(x_1)+ {\mathcal{L}}_{1,\PP}^\ast (c)\omega_1(x_1)  +  \{\mu_{1,\PN,a}^\T q(a_1) + {\mathcal{L}}_{1,\PN}^\ast (c) \} \{1- \omega_1(x_1) \}    ,  
\\
    \vdots   \\
        \E (Y & \mid Z=S=1,x_{2p+2})\\&=   \E (Y   \mid G=11,Z=S=1,x_{2p+2}) \omega_1(x_{2p+2})  + \E (Y   \mid G=01,Z=S=1,x_{2p+2})  \{1- \omega_1(x_{2p+2}) \}\\ &=      {\mathcal{L}}_{{{1,\PP}}}(x_{2p+2})  \omega_1(x_{2p+2}) +  {\mathcal{L}}_{{1,\PN}}(x_{2p+2}) \{1- \omega_1(x_{2p+2}) \}
        \\&=     \mu_{1,\PP,a}^\T q(a_{2p+2})\omega_1(x_{2p+2})+ {\mathcal{L}}_{1,\PP}^\ast (c)\omega_1(x_{2p+2})  +  \{\mu_{1,\PN,a}^\T q(a_{2p+2}) + {\mathcal{L}}_{1,\PN}^\ast (c) \} \{1- \omega_1(x_{2p+2}) \}   .  
    \end{aligned}  
    \end{equation}
For any given $C=c$,  the above system can be rewritten as
\[
\begin{aligned}
    &\begin{pmatrix}
q^\T(a_1)\omega_1(x_1) & \omega_1(x_1)   &q^\T(a_1)\{1-\omega_1(x_1) \} &1-\omega_1(x_1) \\
q^\T(a_2)\omega_1(x_2) & \omega_1(x_2)   &q^\T(a_2)\{1-\omega_1(x_2) \} &1-\omega_2(x_2) \\
\vdots & \ldots   & \ldots& \vdots \\
q^\T(a_{2p+2})\omega_1(x_{2p+2}) & \omega_1(x_{2p+2})   &q^\T(a_{2p+2})\{1-\omega_1(x_{2p+2}) \} &1-\omega_{1}(x_{2p+2}) \\
\end{pmatrix}
\begin{Bmatrix}
 \mu_{1,\PP,a}\\
{\mathcal{L}}_{1,\PP}^\ast (c)\\
\mu_{1,\PN,a}\\
{\mathcal{L}}_{1,\PN}^\ast (c) 
\end{Bmatrix}\\&
=
\begin{Bmatrix}
\E(Y \mid Z= S=1, X=x_1) \\
\E(Y \mid Z= S=1, X=x_2) \\
\vdots\\
\E(Y \mid Z= S=1, X=x_{2p+2})
\end{Bmatrix}.
\end{aligned}
\]

As long as condition~\eqref{eq:iden-exp} is satisfied, the coefficient matrix has full rank, and the system of equations has a unique solution. Consequently, the parameters $\mu_{1,\PP,a}$, ${\mathcal{L}}_{1,\PP}(c)$, $\mu_{1,\PN,a}$, and ${\mathcal{L}}_{1,\PN}(c)$ are identifiable. According to a similar logic, we can also identify $\mu_{0,\PP,a}$, ${\mathcal{L}}_{0,\PP}(c)$, $\mu_{0,\NP,a}$, and ${\mathcal{L}}_{0,\NP}(c)$. It then follows that ${\mathcal{L}}_{z,s_0s_1}(X)$ is identifiable: $
{\mathcal{L}}_{z,s_0s_1}(X) = \mu_{z,s_0s_1,a}^\T q(A) + {\mathcal{L}}_{z,s_0s_1}^\ast(C).$ 
\end{proof}       
\subsection{The identification of Bayesian decision rule}
\begin{lemma}\label{lem:prop-pce}
   Under Assumptions \ref{assumption:RCT}, {\red \ref{assumption:NoInteraction}}, \ref{assump:odds}, and condition \eqref{eq:iden-exp},  the optimal Bayesian decision rule is identifiable.\end{lemma}
\begin{proof}
    After identifying $\pi_{s_0s_1}(X)$ and ${\mathcal{L}}_{z,s_0s_1}(X)$, the misclassification reward can be further identified via the following equation:
    \begin{align*} 
\begin{aligned} {\mathcal{R}}&(\tilde{s}_0\tilde{s}_1\mid s_0s_1)
 = \mathbb{E} \left[ \begin{matrix}
          \{{{\mathcal{L}}}_{1,s_0s_1}(X)- c_{1}(X)\}\pi_{s_0s_1}(X)\rho( {\tilde{s}_0\tilde{s}_1},X) \\+\{{{\mathcal{L}}}_{0,s_0s_1}(X)- c_{0}(X)\}\pi_{s_0s_1}(X)\{1-\rho( {\tilde{s}_0\tilde{s}_1},X)\}  
     \end{matrix} \right] \Bigg/{\mathbb{E}\left\{\pi_{s_0s_1}(X)\right\}},
\end{aligned}
\end{align*}   Consequently, Theorem~\ref{thm:maxim-reward} implies that the optimal Bayesian decision rule is identifiable.
\end{proof}
\section{Estimation performance}
\subsection{Preliminaries}In this section, we present several preliminary lemmas.
\begin{lemma}
\label{lem:ineq} 
Let $\hat{f}$ and $f$ be any real numbers. Then
\[
\left| \mathbb{I}(\hat{f} > 0) - \mathbb{I}(f > 0) \right| \leq \mathbb{I}(|f| \leq |\hat{f} - f|).
\] 
\end{lemma}

\begin{proof}
We first note that
\[
\left| \mathbb{I}(\hat{f} > 0) - \mathbb{I}(f > 0) \right| = \mathbb{I}(\hat{f} \cdot f < 0),
\]
since the indicator values differ if and only if \( \hat{f} \) and \( f \) have opposite signs.

Now suppose \( \hat{f} \cdot f < 0 \). Then
\[
|\hat{f}| + |f| = |\hat{f} - f|,
\]
which implies that \( |f| \leq |\hat{f} - f| \).

Therefore, whenever \( \left| \mathbb{I}(\hat{f} > 0) - \mathbb{I}(f > 0) \right| = 1 \), we must have \( \mathbb{I}(|f| \leq |\hat{f} - f|) = 1 \), and the inequality holds. This completes the proof.
\end{proof} 
\begin{lemma}
\label{lem:abs-small-tail}
Suppose the random variable \( h(X) \) has a density function \( p_h(x) \) that is uniformly bounded on \( \mathbb{R} \); that is, there exists a constant \( C_0 > 0 \) such that
\[
p_h(x) \le C_0 \quad \text{for all } x \in \mathbb{R}.
\]
Then there exists a constant \( C > 0 \) such that, for all \( t > 0 \),
\[
\mathbb{P}\left( \left| h(X) \right| \le t \right) \le C t.
\]
\end{lemma}

\begin{proof}
We have
\[
\mathbb{P}\left( \left| h(X) \right| \le t \right) = \int_{-t}^{t} p_h(x) \, dx \le \int_{-t}^{t} C_0 \, dx = 2C_0 t.
\]
Setting \( C = 2C_0 \) yields the result.
\end{proof}

{\begin{lemma}
\label{lem:pi-rate}
Under Assumptions \ref{assump:odds} and \ref{rate-cond}(i), for any $s_0 s_1 \in \{ 00,01,10,11\}$, the plug-in estimator $\hat{\pi}_{s_0s_1}(X)$ from the first method in Step 1 of Section \ref{sec:est-proc} satisfies
\begin{equation*}
\|\hat{\pi}_{s_0s_1} - \pi_{s_0s_1}\|_{\infty} = O_p(n^{-\gamma_*}),
\end{equation*}
where $\gamma_* = \min\{\gamma_0, \gamma_1\}$.
\end{lemma}
\begin{proof}
By Lemma~\ref{lem:prop-ps}, the function \( \pi_{11}(X) \) can be expressed as
\[
\pi_{11}(X) = 
\begin{cases}
\displaystyle \frac{\theta(X) + \{1 - \theta(X)\}\{e_0(X) + e_1(X)\} - \theta(X) \sqrt{\delta(X)} }{1 - \theta(X)}, & \theta(X) \ne 1, \\
e_0(X)e_1(X), & \theta(X) = 1,
\end{cases}
\]
where
\[
\delta(X) = \left[1 + \left\{\theta^{-1}(X) - 1\right\} \left\{e_0(X) + e_1(X)\right\}\right]^2 
- 4 \theta(X) \left\{\theta^{-1}(X) - 1\right\} e_0(X)e_1(X).
\]

We consider two cases:

\begin{itemize}
    \item[{(i)}]  {If \( \theta(X) = 1 \):} Then
    \[
    \pi_{11}(X) = e_0(X)e_1(X), \quad \hat{\pi}_{11}(X) = \hat{e}_0(X)\hat{e}_1(X),
    \]
    and hence
    \begin{align*}
    |\hat{\pi}_{11}(X) - \pi_{11}(X)| 
    &= |\hat{e}_0(X)\hat{e}_1(X) - e_0(X)e_1(X)| \\
    &\le |\hat{e}_0(X)| \cdot |\hat{e}_1(X) - e_1(X)| + |e_1(X)| \cdot |\hat{e}_0(X) - e_0(X)| \\
    &= O_p(n^{-\gamma_1}) + O_p(n^{-\gamma_0}) = O_p(n^{-\gamma_*}),
    \end{align*}
    where \( \gamma_* = \min\{\gamma_0, \gamma_1\} \).

    \item[{(ii)}]  {If \( \theta(X) \ne 1 \):} Define
    \[
    q(e_0, e_1) = \frac{ \theta + (1 - \theta)(e_0 + e_1) - \theta \sqrt{\delta} }{1 - \theta },
    \]
    where
    \[
    \delta = \left\{1 + (\theta^{-1} - 1)(e_0 + e_1)\right\}^2 
    - 4 \theta(\theta^{-1} - 1)e_0e_1.
    \]
    Then \( \pi_{11}(X) = q(e_0(X), e_1(X)) \) and \( \hat{\pi}_{11}(X) = q(\hat{e}_0(X), \hat{e}_1(X)) \). Since \( q(\cdot,\cdot) \) is continuously differentiable on bounded domains, it is Lipschitz continuous, so
    \[
    |\hat{\pi}_{11}(X) - \pi_{11}(X)| \le C \left( |\hat{e}_0(X) - e_0(X)| + |\hat{e}_1(X) - e_1(X)| \right) = O_p(n^{-\gamma_*}),
    \]
    for some constant \( C > 0 \).
\end{itemize}

Combining the two cases gives the desired result for the proportion \( \pi_{11}(X) \). The same argument applies to the other components \( \pi_{ij}(X) \) for \( (i, j) \in \{0,1\}^2 \), since their definitions are also continuously differentiable functions of \( e_0(X) \), \( e_1(X) \) and $\pi_{11}(X)$.
\end{proof}

}
\subsection{The  rate of $\mathcal{R}(\tilde{s}_0\tilde{s}_1\mid s_0s_1)$  in Theorem \ref{thm:rate}}
To prove Theorem \ref{thm:rate}, we first present a more general result under a set of assumptions that parallel those in Theorem \ref{thm:rate}.
\begin{lemma}[Convergence Rate]{Assume the following conditions hold:
\label{lem:rate-lem}
\begin{itemize}
    \item[] (i) The functions $f_1(x), f_2(x), f_3(x)$ and their estimators $\hat{f}_1(x), \hat{f}_2(x), \hat{f}_3(x)$ are uniformly bounded by some constant $M > 0$, i.e.,
  $$
  \sup_x |f_\ell(x)| \le M,\quad \sup_x |\hat{f}_\ell(x)| \le M,\quad \text{for all } \ell = 1,2,3.
  $$
\item[] (ii)  For all $\ell=1,2,3$, we have $ 
\|\hat{f}_\ell - f_\ell\|_\infty = O_p(n^{-\alpha_\ell})$ for some $\alpha_\ell>0.$
\item[] (iii)  The function $ f_1(x)$ has a density that is uniformly bounded over $x \in \mathbb{R}^m$. 
 \end{itemize}
  Then define
$$
\mathcal{C}_n = \frac{\mathbb{P}_n \left[ \hat f_2(X) \cdot  \hat f_3(X) \cdot\max\{ \hat f_1(X), 0 \} \right]}{\mathbb{P}_n \left\{ \hat f_3(X) \right\}},\quad
\mathcal{C}  = \frac{\mathbb{E}  \left[   f_2(X) \cdot   f_3(X) \cdot\max\{ f_1(X),0 \} \right]}{\mathbb{E}  \left\{  f_3(X) \right\}},
$$
we have
$$
|\mathcal{C}_n - \mathcal{C}| \leq O_p\left( \|\hat{f}_1 - f_1\|_\infty \vee \|\hat{f}_2 - f_2\|_\infty \vee \|\hat{f}_3 - f_3\|_\infty \vee n^{-1/2} \right).
$$
}
\end{lemma}

\begin{proof}
Define
\[\begin{gathered}
    \mathcal{A}_n=  \mathbb{P}_n \left[ \hat f_2(X) \cdot \hat f_3(X) \cdot \max\{ \hat f_1(X),0 \} \right] ,\quad
    \mathcal{B}_n=  \mathbb{P}_n \left\{ \hat f_3(X) \right\}, \\
    \mathcal{A}  =  \mathbb{E}  \left[   f_2(X) \cdot  f_3(X) \cdot \max\{ f_1(X),0 \} \right] ,\quad
    \mathcal{B}  =  \mathbb{E} \left[ f_3(X)  \right].
\end{gathered}\]
By Slutsky's theorem, if \( \mathcal{B}_n \xrightarrow{p} \mathcal{B} > 0 \), then
\[
\mathcal{C}_n - \mathcal{C} = \frac{\mathcal{A}_n}{\mathcal{B}_n} - \frac{\mathcal{A}}{\mathcal{B}} = \frac{\mathcal{A}_n \mathcal{B} - \mathcal{A} \mathcal{B}_n}{\mathcal{B}^2} + o_p(1).
\]
We decompose
\[\mathcal{A}_n \mathcal{B} - \mathcal{A} \mathcal{B}_n = (\mathcal{A}_n - \mathcal{A})\mathcal{B} - \mathcal{A}(\mathcal{B}_n - \mathcal{B}).\]
Let
\begin{align*}
\mathcal{A}_{1n} &= \mathbb{P}_n\left\{\hat{f}_2(X) \cdot \hat{f}_3(X)  \cdot \mathbb{I}(\hat f_1(X)> 0)\right\} - \mathbb{P}_n\left\{f_2(X) \cdot f_3(X)\cdot \mathbb{I}(\hat f_1(X)> 0)\right\}, \\
\mathcal{A}_{2n} &= \mathbb{P}_n\left\{f_2(X) \cdot f_3(X)  \cdot \mathbb{I}(\hat f_1(X)> 0)\right\} - \mathbb{P}_n\left\{f_2(X) \cdot f_3(X)\cdot \mathbb{I}(  f_1(X)> 0)\right\}, \\
\mathcal{A}_{3n} &= \mathbb{P}_n\left\{f_2(X) \cdot f_3(X)\cdot \mathbb{I}(  f_1(X)> 0)\right\} - \mathbb{E}\left\{f_2(X) \cdot f_3(X)\cdot \mathbb{I}(  f_1(X)> 0)\right\}.
\end{align*}
For $\mathcal{A}_{1n}$, we have that
\[
\begin{aligned}
     |\mathcal{A}_{1n}| &\leq \frac{1}{n} \sum_{i=1}^n \left| \hat{f}_2(X_i) \hat{f}_3(X_i) - f_2(X_i) {f}_3(X_i)\right|\\& \leq \frac{1}{n} \sum_{i=1}^n \left(\hat{f}_2(X_i)\left|  \hat{f}_3(X_i) -  {f}_3(X_i)\right|+ {f}_3(X_i)\left|  \hat{f}_2(X_i) -  {f}_2(X_i)\right|\right)\\&=O_p\left( \sup_x |\hat{f}_3(x) - f_3(x)| +\sup_x |\hat{f}_2(x) - f_2(x)|  \right)
\end{aligned}
\]
For $\mathcal{A}_{2n}$, we have that 
\begin{align*} |\mathcal{A}_{2n}| 
&\leq \frac{1}{n} \sum_{i=1}^n f_2(X_i) \cdot f_3(X_i) \cdot \left|  \mathbb{I}( \hat  f_1(X_i)> 0) -  \mathbb{I}(  f_1(X_i)> 0) \right| \\&\leq \sup_x   f_2(x)\cdot\sup_x  f_3(x) \cdot\frac{1}{n} \sum_{i=1}^n \left|  \mathbb{I}( \hat  f_1(X_i)> 0) -  \mathbb{I}(  f_1(X_i)> 0) \right| \\&\leq \sup_x   f_2(x)\cdot\sup_x  f_3(x) \cdot\frac{1}{n} \sum_{i=1}^n \left|  \mathbb{I}(f_1(X_i)\leq \vert\hat  f_1(X_i)-  f_1(X_i)\vert)   \right| \\
&=  M^2\cdot   \pr\left(  f_1(X_i)\leq \vert\hat  f_1(X_i)-  f_1(X_i)\vert\right) +o_p(1) \\
&\lesssim O_p\left( \sup_x |\hat{f}_1(x) - f_1(x)|   \right),
\end{align*}    
where the third inequality follows from Lemma~\ref{lem:ineq}, and the last equality holds because Lemma \ref{lem:abs-small-tail}.
By the central limit theorem, we have 
 $$\mathcal{A}_{3n} =  \frac{1}{n} \sum_{i=1}^n f_2(X_i) \cdot{f_3}(X_i)  \cdot \mathbb{I}(  f_1(X_i)> 0)-  {\mathbb{E} \left[ f_2(X) f_3(X) \cdot \mathbb{I}(  f_1(X)> 0) \right]}=O_p(n^{-1/2})$$
 We then have \[
\begin{aligned}
    |(\mathcal{A}_n - \mathcal{A})\mathcal{B}|& = |\mathcal{A}_{1n} + \mathcal{A}_{2n} + \mathcal{A}_{3n} ||\mathcal{B}|\\& =O_p\left(   ||\hat{f}_1(x) - f_1(x)||_{\infty} +   ||\hat{f}_3(x) - f_3(x)||_{\infty} +  ||\hat{f}_2(x) - f_2(x)||_{\infty} +n^{-1/2}\right) 
\end{aligned}
\]
Furthermore, write $$
\mathcal{B}_n - \mathcal{B} = \mathbb{P}_n \{\hat{f}_3(X) - f_3(X)\} + \big[\mathbb{P}_n \{f_3(X)\} - \mathbb{E}\{f_3(X)\}\big].$$ 
The first term is bounded by
\[
|\mathbb{P}_n \{\hat{f}_3(X) - f_3(X)\}| \leq \|\hat{f}_3 (x)- f_3(x)\|_\infty,
\]
and the second term is \(O_p(n^{-1/2})\) by the central limit theorem. 
Multiplying by a bounded \(\mathcal{A}\), we get
\[
|(\mathcal{B}_n - \mathcal{B}) \mathcal{A}| = O_p\big(\|\hat{f}_3 (x)- f_3(x)\|_\infty + n^{-1/2}\big).
\]
 
Therefore,
\[
|(\mathcal{A}_n - \mathcal{A})\mathcal{B}| = O_p\left(\|\hat{f}_1(x) - f_1(x)\|_{\infty} + \|\hat{f}_2(x) - f_2(x)\|_{\infty} + \|\hat{f}_3(x) - f_3(x)\|_{\infty} + n^{-1/2}\right).
\]
Combining the above,
\[
\begin{aligned}
    |\mathcal{C}_n - \mathcal{C}| &= \Bigg| \frac{\mathcal{A}_n \mathcal{B} - \mathcal{A} \mathcal{B}_n}{\mathcal{B}^2}\Bigg| + o_p(1)\\&=\Bigg| \frac{(\mathcal{A}_n - \mathcal{A})\mathcal{B}  }{\mathcal{B}^2}\Bigg| +\Bigg| \frac{  \mathcal{A}(\mathcal{B}_n - \mathcal{B})}{\mathcal{B}^2}\Bigg| + o_p(1)\\& \le O_p\left(\|\hat{f}_1(x) - f_1(x)\|_{\infty} \vee \|\hat{f}_2(x) - f_2(x)\|_{\infty} \vee \|\hat{f}_3(x) - f_3(x)\|_{\infty} \vee n^{-1/2}\right).
\end{aligned}
\]
\end{proof}
\subsection{The proof of Theorem \ref{thm:rate}}
\begin{proof}
    
  Let $  \hat\rho( {\tilde{s}_0\tilde{s}_1},X) = \mathbb{I}\left\{ \hat{{\mathcal{L}}}_{1,\tilde{s}_0\tilde{s}_1}(X) -  c_1(X) \geq \hat{{\mathcal{L}}}_{0,\tilde{s}_0\tilde{s}_1}(X) -  c_0(X) \right\}.$ We have,  \begin{align*}
\begin{aligned} \hat{\mathcal{R}}(\tilde{s}_0\tilde{s}_1\mid s_0s_1)&
 = \dfrac{\mathbb{P}_n \left[ \begin{matrix}
          \{\hat{{\mathcal{L}}}_{1,s_0s_1}(X) -  c_1(X) \}\hat\pi_{s_0s_1}(X)\hat\rho( {\tilde{s}_0\tilde{s}_1},X) \\+\{\hat{{\mathcal{L}}}_{0,s_0s_1}(X) -  c_0(X) \}\hat\pi_{s_0s_1}(X)\{1-\hat\rho( {\tilde{s}_0\tilde{s}_1},X)\}  
     \end{matrix} \right]} {\mathbb{P}_n\left\{\hat\pi_{s_0s_1}(X)\right\}}\\&
 =  \dfrac{\mathbb{P}_n \left[ \begin{matrix}
          \{\hat{{\mathcal{L}}}_{1,s_0s_1}(X)-  c_1(X)\}\hat\pi_{s_0s_1}(X)\mathbb{I}\left( \hat{{\mathcal{L}}}_{1,\tilde{s}_0\tilde{s}_1}(X) -  c_1(X) - \hat{{\mathcal{L}}}_{0,\tilde{s}_0\tilde{s}_1}(X) +  c_ 0(X) >0\right)  
     \end{matrix} \right]} {\mathbb{P}_n\left\{\hat\pi_{s_0s_1}(X)\right\}}\\&~~~~~+  \dfrac{\mathbb{P}_n \left[ \begin{matrix}
          \{\hat{{\mathcal{L}}}_{0,s_0s_1}(X)- c_0(X)\}\hat\pi_{s_0s_1}(X)\mathbb{I}\left( \hat{{\mathcal{L}}}_{1,\tilde{s}_0\tilde{s}_1}(X) -  c_1(X) - \hat{{\mathcal{L}}}_{0,\tilde{s}_0\tilde{s}_1}(X) +  c_0(X)<0\right)  
     \end{matrix} \right]} {\mathbb{P}_n\left\{\hat\pi_{s_0s_1}(X)\right\}}.
\end{aligned} 
\end{align*} 

Consequently, Lemma~\ref{lem:pi-rate} provides the convergence rate for $\hat{\pi}_{s_0s_1}(X)$ when we employ the first estimation approach in Step 1 of Section~\ref{sec:est-proc} to estimate $\pi_{s_0s_1}(X)$. Under correct parametric specification, the second estimation approach in Step 1  achieves a root-$n$ convergence rate. 

The conclusion then follows directly from Lemma~\ref{lem:rate-lem} by setting:
$$
\begin{gathered}
f_1(X) = {\mathcal{L}}_{1,\tilde{s}_0\tilde{s}_1}(X) - {\mathcal{L}}_{0,\tilde{s}_0\tilde{s}_1}(X) - c_1(X) + c_0(X), \\
f_2(X) = \pi_{s_0s_1}(X), \quad f_3 (X) = {\mathcal{L}}_{1,s_0s_1}(X) - c_1(X) .
\end{gathered}
$$
and 
$$
\begin{gathered}
f_1(X) = {\mathcal{L}}_{1,\tilde{s}_0\tilde{s}_1}(X) - {\mathcal{L}}_{0,\tilde{s}_0\tilde{s}_1}(X) - c_1(X) + c_0(X), \\
f_2(X) = \pi_{s_0s_1}(X),  \quad f_3 (X) = {\mathcal{L}}_{0,s_0s_1}(X) - c_0(X).
\end{gathered}
$$
\end{proof}
{  
\section{The proof of Proposition  \ref{prop:selection}}
\subsection{The selection consistency  in Proposition  \ref{prop:selection}}
{
Define $$\begin{gathered}
    {\hat{\mathcal{Q}}}_{s_0s_1}(x)=\sum_{\tilde{s}_0 ,\tilde{s}_1\in\{0,1\} } {\hat{\mathcal{R}}}(\tilde{s}_0\tilde{s}_1\mid s_0s_1) \cdot \hat\pi_{\tilde{s}_0 \tilde{s}_1}(x) ,~~{{\mathcal{Q}}}_{s_0s_1}(x)=\sum_{\tilde{s}_0 ,\tilde{s}_1\in\{0,1\} } {{\mathcal{R}}}(\tilde{s}_0\tilde{s}_1\mid s_0s_1) \cdot \pi_{\tilde{s}_0 \tilde{s}_1}(x)  .
\end{gathered}$$
\begin{lemma}
    \label{lem:Q-consistency}
Under Assumption~\ref{rate-cond}, for any fixed $X=x $ and $  s_0 ,s_1\in\{0,1\}$, the estimated value function $\hat{Q}_{s_0 s_1}(x)$ is a consistent estimator of $Q_{s_0 s_1}(x)$, i.e.,
$ 
\hat{Q}_{s_0 s_1}(x) \xrightarrow{p} Q_{s_0 s_1}(x). $
\end{lemma}
\begin{proof}
According to Theorem~\ref{thm:rate}, for any \(s_0,s_1, \tilde{s}_0, \tilde{s}_1 \in \{0,1\} \), we have
\[
\hat{\mathcal{R}}(\tilde{s}_0 \tilde{s}_1 \mid s_0 s_1) \xrightarrow{p} \mathcal{R}(\tilde{s}_0 \tilde{s}_1 \mid s_0 s_1). 
\] 

By definition,
\[
\hat{Q}_{s_0 s_1}(x) = \sum_{\tilde{s}_0, \tilde{s}_1} \hat{\mathcal{R}}(\tilde{s}_0 \tilde{s}_1 \mid s_0 s_1) \cdot \hat{\pi}_{\tilde{s}_0 \tilde{s}_1}(x),
\quad
Q_{s_0 s_1}(x) = \sum_{\tilde{s}_0, \tilde{s}_1} \mathcal{R}(\tilde{s}_0 \tilde{s}_1 \mid s_0 s_1) \cdot \pi_{\tilde{s}_0 \tilde{s}_1}(x).
\]

Then by the triangle inequality,
\[
\left| \hat{Q}_{s_0 s_1}(x) - Q_{s_0 s_1}(x) \right|
\leq \sum_{\tilde{s}_0, \tilde{s}_1}
\left| \hat{\mathcal{R}}(\tilde{s}_0 \tilde{s}_1 \mid s_0 s_1) \hat{\pi}_{\tilde{s}_0 \tilde{s}_1}(x)
- \mathcal{R}(\tilde{s}_0 \tilde{s}_1 \mid s_0 s_1) \pi_{\tilde{s}_0 \tilde{s}_1}(x) \right|.
\]

Each term can be bounded as
\[
\begin{aligned}
 \left| \hat{\mathcal{R}}(\tilde{s}_0 \tilde{s}_1 \mid s_0 s_1) \hat{\pi}_{\tilde{s}_0 \tilde{s}_1}(x)
- \mathcal{R}(\tilde{s}_0 \tilde{s}_1 \mid s_0 s_1) \pi_{\tilde{s}_0 \tilde{s}_1}(x) \right| 
&\leq
\left| \hat{\mathcal{R}}(\tilde{s}_0 \tilde{s}_1 \mid s_0 s_1) - \mathcal{R}(\tilde{s}_0 \tilde{s}_1 \mid s_0 s_1) \right| \cdot \left| \hat{\pi}_{\tilde{s}_0 \tilde{s}_1}(x) \right| \\
&\quad +
\left| \mathcal{R}(\tilde{s}_0 \tilde{s}_1 \mid s_0 s_1) \right| \cdot \left| \hat{\pi}_{\tilde{s}_0 \tilde{s}_1}(x) - \pi_{\tilde{s}_0 \tilde{s}_1}(x) \right|.
\end{aligned}
\]

Under the boundedness assumption \( |\hat{\pi}_{\tilde{s}_0 \tilde{s}_1}(x)| \leq M \), it follows that
\[
\left| \hat{Q}_{s_0 s_1}(x) - Q_{s_0 s_1}(x) \right| \leq C \sum_{\tilde{s}_0, \tilde{s}_1}
\left(
\left| \hat{\mathcal{R}} (\tilde{s}_0 \tilde{s}_1 \mid s_0 s_1)- \mathcal{R} (\tilde{s}_0 \tilde{s}_1 \mid s_0 s_1)\right| + \left| \hat{\pi} (x) - \pi(x)  \right|
\right),
\]
for some constant \( C > 0 \), where all terms vanish in probability as \( n \to \infty \).

Therefore,
\[
\hat{Q}_{s_0 s_1}(x) \xrightarrow{p} Q_{s_0 s_1}(x),
\]
which establishes the consistency of \(\hat{Q}_{s_0 s_1}(x)\) at any fixed \(x\).
\end{proof}

}

\subsection{The proof of Proposition \ref{prop:selection}}
\begin{proof}
Fix any covariate value $X = x$. Recall that the true optimal decision set is defined as

$$\mathcal{S}(x) = \arg\max_{s_0s_1 }  \sum_{\tilde{s}_0 ,\tilde{s}_1\in\{0,1\} }{ {\mathcal{R}}}(\tilde{s}_0\tilde{s}_1\mid s_0s_1) \cdot \pi_{\tilde{s}_0 \tilde{s}_1}(x)$$

We aim to show that the estimated optimal decision $\hat {G}^\ast$, which maximizes the estimated reward $\hat{\mathcal{Q}}_{s_0s_1}(x)$, lies in $\mathcal{S}(x)$ with high probability as $n \to \infty$; that is,

$$
\pr\left( \hat {G}^\ast \notin \mathcal{S}(x) \mid X=x\right) \to 0.
$$

According to Lemma \ref{lem:Q-consistency}, we know that   $\hat{\mathcal{Q}}_{s_0s_1}(x)$ is a consistent estimator of $\mathcal{Q}_{s_0s_1}(x)$ for every $s_0s_1 \in \{00,01,10,11\}$. Therefore,

$$
\max_{s_0s_1 \in \{00,01,10,11\}} \left| \hat{\mathcal{Q}}_{s_0s_1}(x) - \mathcal{Q}_{s_0s_1}(x) \right| \xrightarrow{p} 0.
$$

Let $\delta > 0$ denote the smallest positive gap in reward between any suboptimal decision and the optimal set. More precisely, for any $s_0s_1 \notin \mathcal{S}(x)$, there exists some $\bar{s}_0\bar{s}_1 \in \mathcal{S}(x)$ such that

$$
\mathcal{Q}_{\bar{s}_0\bar{s}_1}(x) - \mathcal{Q}_{s_0s_1}(x) \ge \delta.
$$

Choose $\xi = \delta/2 > 0$. By consistency, there exists an integer $N > 0$ such that for all $n > N$,

$$
\max_{s_0s_1 \in \{00,01,10,11\}} \left| \hat{\mathcal{Q}}_{s_0s_1}(x) - \mathcal{Q}_{s_0s_1}(x) \right| < \xi \quad \text{with probability at least } 1 - \varepsilon.
$$

Now consider any suboptimal strategy $s_0s_1 \notin \mathcal{S}(x)$ and any optimal strategy $\bar{s}_0\bar{s}_1 \in \mathcal{S}(x)$. We have:

$$
\hat{\mathcal{Q}}_{s_0s_1}(x) < \mathcal{Q}_{s_0s_1}(x) + \xi 
< \mathcal{Q}_{\bar{s}_0\bar{s}_1}(x) - \delta + \xi 
= \mathcal{Q}_{\bar{s}_0\bar{s}_1}(x) - \xi 
< \hat{\mathcal{Q}}_{\bar{s}_0\bar{s}_1}(x).
$$

This implies that every suboptimal strategy $s_0s_1$ is assigned strictly lower estimated reward than at least one optimal strategy $\bar{s}_0\bar{s}_1$, so $\hat {G}^\ast\not= s_0s_1$, i.e., it must be that $\hat {G}^\ast \in \mathcal{S}(x)$. Therefore,

$$
\pr\left( \hat {G}^\ast \notin \mathcal{S}(x) \mid X=x\right) 
\le \pr\left( \max_{s_0s_1 \in \{00,01,10,11\}} \left| \hat{\mathcal{Q}}_{s_0s_1}(x) - \mathcal{Q}_{s_0s_1}(x) \right| \geq \xi~ \Big|~X=x \right)
\to 0 ~ \text{as } n \to \infty.
$$

This completes the proof.
\end{proof}

\section{Estimation method for Logistic models}
\label{sec:logistic-models}
\subsection{Verification of $\theta(X)$ in Section \ref{sec:ER-sim}}\label{ssec:verification-thetaX} 
According to \eqref{sensitivity:logistic} in the main text, we have
\begin{equation*}
\Pr(G=s_0 s_1 \mid X) =\frac{\exp\{s_1(\iota_{01} + \iota_{01,X}^\T X)\}\exp\{s_0(\iota_{10} + \iota_{10,X}^\T X)\}\exp(\eta s_0s_1)}{\sum_{\tilde{s}_0,\tilde{s}_1\in\{0,1\}}\exp\{\tilde{s}_1(\iota_{01} + \iota_{01,X}^\T X)\}\exp\{\tilde{s}_0(\iota_{10} + \iota_{10,X}^\T X)\}\exp(\eta \tilde{s}_0\tilde{s}_1)},
\end{equation*}
specifically, 
\begin{align*}
\Pr(G=00 \mid X) &= \frac{1}{1 + \exp(\iota_{01} + \iota_{01,X}^\T X) + \exp(\iota_{10} + \iota_{10,X}^\T X) + \exp(\iota_{01} + \iota_{01,X}^\T X + \iota_{10} + \iota_{10,X}^\T X + \eta)}, \\
\Pr(G=01\mid X) &= \frac{\exp(\iota_{01} + \iota_{01,X}^\T X)}{1 + \exp(\iota_{01} + \iota_{01,X}^\T X) + \exp(\iota_{10} + \iota_{10,X}^\T X) + \exp(\iota_{01} + \iota_{01,X}^\T X + \iota_{10} + \iota_{10,X}^\T X + \eta)}, \\
\Pr(G=10 \mid X) &= \frac{\exp(\iota_{10} + \iota_{10,X}^\T X)}{1 + \exp(\iota_{01} + \iota_{01,X}^\T X) + \exp(\iota_{10} + \iota_{10,X}^\T X) + \exp(\iota_{01} + \iota_{01,X}^\T X + \iota_{10} + \iota_{10,X}^\T X + \eta)}, \\
\Pr(G=11 \mid X) &= \frac{\exp(\iota_{01} + \iota_{01,X}^\T X + \iota_{10} + \iota_{10,X}^\T X + \eta)}{1 + \exp(\iota_{01} + \iota_{01,X}^\T X) + \exp(\iota_{10} + \iota_{10,X}^\T X) + \exp(\iota_{01} + \iota_{01,X}^\T X + \iota_{10} + \iota_{10,X}^\T X + \eta)}.
\end{align*}

By verifying the definition of the odds function, we know that 
$$\theta(X) = \frac{\pi_{00}(X)\pi_{11}(X)}{\pi_{10}(X)\pi_{01}(X)} = \exp(\eta).$$

\subsection{EM algorithm}
\label{EMalgorithm}In this section, we provide an EM algorithm for estimating the strata probabilities $\pi_{s_0s_1}(X)$ under the logistic model. 
Although we cannot fully observe \( G \), we can use the EM algorithm to find the MLEs by treating \( G \) as missing data.


In the E-step, we compute the conditional probabilities $\Pr(G_i = s_0s_1 \mid -)$ for each observation $(Z_i, S_i, X_i)$ given current parameter estimates $(\hat{\iota}_0^k, \hat{\iota}_1^k)$ as follows:

If $(Z_i = 1, S_i = 1)$:
\[
\Pr(G_i = 01 \mid -) = \frac{{\exp({\hat \iota}^{k}_{01} + {\hat \iota}^{k}_{01,X}  X)}}{{\exp({\hat \iota}^{k}_{01} + {\hat \iota}^{k}_{01,X}  X)}+\exp({\hat \iota}^{k}_{01} + {\hat \iota}^{k}_{01,X}  X + {\hat \iota}^{k}_{10} + {\hat \iota}^{k\T}_{10,X} X + \eta)} ,\]\[ 
    \Pr(G_i = 11 \mid -) = \frac{{\exp({\hat \iota}^{k}_{01} + {\hat \iota}^{k}_{01,X}  X + {\hat \iota}^{k}_{10} + {\hat \iota}^{k\T}_{10,X} X + \eta)}}{{\exp({\hat \iota}^{k}_{01} + {\hat \iota}^{k}_{01,X}  X)}+\exp({\hat \iota}^{k}_{01} + {\hat \iota}^{k}_{01,X}  X + {\hat \iota}^{k}_{10} + {\hat \iota}^{k\T}_{10,X} X + \eta)}.
\]

If $(Z_i = 0, S_i = 1)$:
\[
\Pr(G_i = 1 0\mid -) = \frac{{\exp({\hat \iota}^{k}_{10} + {\hat \iota}^{k\T}_{10,X} X)}}{{\exp({\hat \iota}^{k}_{10} + {\hat \iota}^{k\T}_{10,X} X)}+\exp({\hat \iota}^{k}_{01} + {\hat \iota}^{k}_{01,X}  X + {\hat \iota}^{k}_{10} + {\hat \iota}^{k\T}_{10,X} X + \eta)}, \]\[ 
    \Pr(G_i = 11 \mid -) = \frac{{\exp({\hat \iota}^{k}_{01} + {\hat \iota}^{k}_{01,X}  X + {\hat \iota}^{k}_{10} + {\hat \iota}^{k\T}_{10,X} X + \eta)}}{{\exp({\hat \iota}^{k}_{10} + {\hat \iota}^{k\T}_{10,X} X)}+\exp({\hat \iota}^{k}_{01} + {\hat \iota}^{k}_{01,X}  X + {\hat \iota}^{k}_{10} + {\hat \iota}^{k\T}_{10,X} X + \eta)} .
\]

If $(Z_i = 1, S_i = 0)$:
\[
\Pr(G_i = 1 0\mid -) = \frac{{\exp({\hat \iota}^{k}_{10} + {\hat \iota}^{k\T}_{10,X} X)}}{{1+\exp({\hat \iota}^{k}_{10} + {\hat \iota}^{k\T}_{10,X} X)} } ,~~
    \Pr(G_i = 11 \mid -) = \frac{1}{{1+\exp({\hat \iota}^{k}_{10} + {\hat \iota}^{k\T}_{10,X} X)} } .
\]

If $(Z_i = 0, S_i = 0)$:
\[
\Pr(G_i = 1 0\mid -) = \frac{{\exp({\hat \iota}^{k}_{01} + {\hat \iota}^{k}_{01,X}  X)}}{{1+\exp({\hat \iota}^{k}_{01} + {\hat \iota}^{k}_{01,X}  X)} } ,~~
    \Pr(G_i = 11 \mid -) = \frac{1}{{1+\exp({\hat \iota}^{k}_{01} + {\hat \iota}^{k}_{01,X}  X)} } .
\]
Using these probabilities, we construct a weighted pseudo-dataset. For each individual $i$, two pseudo-observations $(G_i = s_0s_1, X_i)$ are generated with the conditional posterior probabilities.

In the M-step, we update the parameters $(\hat{\iota}_0^{k+1}, \hat{\iota}_1^{k+1})$ by fitting two logistic regression models using the weighted samples obtained in the E-step:
$$
\Pr(G = 0s_1 \mid Z = 0, S = 0, X) \propto \exp\left\{ s_1(\iota_{01} + \iota_{01,X}^\T X) \right\},
$$
$$
\Pr(G = s_00 \mid Z = 1, S = 0, X) \propto \exp\left\{ s_0(\iota_{10} + \iota_{10,X}^\T X) \right\}.
$$

\section{Additional discussion about convergence rate}
\label{discussion-convengenc}
In this section, we provide further discussion on the convergence rate, particularly focusing on Assumption \ref{rate-cond}. We expand on the implications of this assumption for the consistency of the estimators and the overall estimation process. 
In Step 1 of the estimation procedure, $\hat{\pi}_{s_0s_1}(X)$ under the first estimation approach typically attains the rate $O_p(n^{-1/2})$ when $\hat{e}_{s_0s_1}(X)$ is correctly specified parametrically, and can also achieve this rate under nonparametric estimation given sufficient regularity \citep{Hirano2003,Zhang2016np}. Under the second estimation approach of Step 1, correct parametric specification of $\pi_{s_0s_1}(X)$ is required to achieve the $O_p(n^{-1/2})$ rate. As outlined in Step 2, the convergence rate of the outcome model $\hat{{\mathcal{L}}}_{z,s_0 s_1}(X)$ depends on that of  $\hat{\pi}_{s_0s_1}(X)$. To ensure that the overall estimator achieves the desired convergence rate of $O_p(n^{-1/2})$, it is necessary for $\hat{{\mathcal{L}}}_{z,s_0 s_1}(X)$ to converge at least at the same rate. This can be guaranteed under correct parametric specification. Under stronger smoothness assumptions and regularity conditions (e.g., those in \citet{AiChen2003,chen2012estimation,ChenPouzo2015}), nonparametric approaches may also yield comparable rates. 
\section{Additional results for simulation studies}
\label{section:linear-model-simulation}
{In this section, we provide additional simulation studies to supplement Section~\ref{sec:experiments}. All settings remain the same as in the main text, except that the outcome model \eqref{eq:exp-models} is replaced with a linear specification.

The baseline covariates $X=\{A,C\}$ are generated from a multivariate normal distribution: $A \sim N(-0.25, 1)$ and $C \sim N(0.25, 1)$. The treatment variable $Z$ is generated according to a logistic model:
\[
\Pr(Z=1 \mid A, C) = \mathrm{expit}(0.15 - \delta A - \delta C),
\]
where $\delta \in \{0, 1\}$ controls the treatment assignment mechanism. When $\delta = 1$, treatment assignment depends on covariates $A$ and $C$; when $\delta = 0$, treatment is randomized.  We specify a multinomial logistic model for the true strata $G \in \{\NN, \NP, \PN, \PP\}$, with $G=\NN$ as the baseline category. The corresponding probabilities are:
\[
\begin{gathered}
\Pr(G = \NN \mid X) \propto 1, \quad \Pr(G = \NP \mid X) \propto \exp(-0.3 - A + 0.5C), \\
\Pr(G = \PN \mid X) \propto \exp(0.4 + A - C), \quad
\Pr(G = \PP \mid X) \propto \exp(0.1 - 0.5C + \eta),
\end{gathered}
\]
where $\eta$ is a parameter varied across scenarios. The observed intermediate variable $S$ is then determined by the consistency assumption.  The potential outcomes $Y_0$ and $Y_1$ are generated from normal distributions with stratum-specific mean functions. Let $\varepsilon_0, \varepsilon_1 \stackrel{iid}{\sim} N(0, 1)$. The potential outcomes are specified as:
\begin{itemize}
\item For $G = \NN$: $Y_0 = Y_1 = 0$;

\item For $G = \NP$: $Y_0 = 5.3 - 1.1C + 1.5A - 1.2C^2 + \varepsilon_0$, $Y_1 = 0$;

\item For $G = \PN$: $Y_0 = 0$, $Y_1 = 7.0 + 1.15C - 1.25A + 1.15C^2 + \varepsilon_1$;

\item For $G = \PP$: $Y_0 = 6 + 1.2C + 1.4A + 1.4C^2 + \varepsilon_0$, 
$Y_1 = 6.5 + 1.2C + 1.4A - 1.25C^2 + \varepsilon_1$.
\end{itemize}
The observed outcome $Y$ is then determined by the consistency assumption: $Y = ZY_1 + (1-Z)Y_0$.
 \subsection{Simulation Results: Principal Causal Effects}

\begin{table}[t]
\caption{Additional simulation results for causal estimands ${\mathcal{L}}_{1,\AT}$, ${\mathcal{L}}_{1,\PN}$, ${\mathcal{L}}_{0,\AT}$, and ${\mathcal{L}}_{0,\NP}$, with bias (\(\times 100\)) and standard error (\(\times 100\)) reported.}
\label{tab: result-simulation-linear-2} 
\centering
\resizebox{0.9965878\columnwidth}{!}{%
\begin{tabular}{ccccccccccc}
\toprule
                        &  & \multicolumn{4}{c}{$\delta=0,\eta=0$}                     &  & \multicolumn{4}{c}{$\delta=0,\eta=0.25$}                  \\ \cline{3-6} \cline{8-11} 
Sample Size             &  & $n=$ 500     & $n=$ 2000    & $n=$ 5000    & $n=$ 20000   &  & $n=$ 500     & $n=$ 2000    & $n=$ 5000    & $n=$ 20000   \\
${\mathcal{L}}_{1,\AT}$ &  & -0.03 (5.84) & 0.13 (2.58)  & 0.11 (1.53)  & 0.04 (0.75)  &  & 0.3 (4.6)    & 0.02 (2.04)  & 0.02 (1.22)  & 0.03 (0.6)   \\
${\mathcal{L}}_{1,\PN}$ &  & 0.04 (3.88)  & -0.08 (1.7)  & -0.06 (1)    & -0.02 (0.49) &  & -0.24 (3.89) & 0.00 (1.73)  & -0.02 (1.03) & -0.02 (0.51) \\
${\mathcal{L}}_{0,\AT}$ &  & 0.4 (4.05)   & 0.07 (1.94)  & 0.06 (1.16)  & -0.08 (0.59) &  & 0.15 (3.64)  & -0.09 (1.76) & -0.1 (0.96)  & -0.05 (0.54) \\
${\mathcal{L}}_{0,\NP}$ &  & -0.22 (2.48) & -0.05 (1.23) & -0.03 (0.73) & 0.04 (0.37)  &  & -0.09 (2.81) & 0.07 (1.38)  & 0.08 (0.77)  & 0.03 (0.42)  \\ \toprule
                        &  & \multicolumn{4}{c}{$\delta=0.25,\eta=0$}                  &  & \multicolumn{4}{c}{$\delta=0.25,\eta=0.25$}               \\ \cline{3-6} \cline{8-11} 
Sample Size             &  & $n=$ 500     & $n=$ 2000    & $n=$ 5000    & $n=$ 20000   &  & $n=$ 500     & $n=$ 2000    & $n=$ 5000    & $n=$ 20000   \\
${\mathcal{L}}_{1,\AT}$ &  & -0.16 (5.83) & -0.17 (2.54) & 0.07 (1.63)  & -0.06 (0.75) &  & 0.04 (4.83)  & -0.15 (2.19) & 0.04 (1.31)  & -0.05 (0.6)  \\
${\mathcal{L}}_{1,\PN}$ &  & 0.13 (3.87)  & 0.12 (1.66)  & -0.04 (1.09) & 0.04 (0.49)  &  & -0.02 (4.07) & 0.13 (1.83)  & -0.03 (1.11) & 0.05 (0.51)  \\
${\mathcal{L}}_{0,\AT}$ &  & 0.46 (4.5)   & 0.22 (2.16)  & -0.14 (1.25) & -0.05 (0.65) &  & 0.49 (4.08)  & 0.05 (1.9)   & -0.07 (1.09) & -0.04 (0.57) \\
${\mathcal{L}}_{0,\NP}$ &  & -0.28 (2.77) & -0.13 (1.33) & 0.07 (0.77)  & 0.02 (0.39)  &  & -0.39 (3.17) & -0.03 (1.47) & 0.05 (0.86)  & 0.02 (0.44)  \\ \toprule
\end{tabular}
}
\end{table}

Table~\ref{tab: result-simulation-linear-2} presents additional simulation results for the marginal expected outcomes ${\mathcal{L}}_{z,s_0s_1}$ under a linear outcome model specification, where $(z, s_0s_1) \in \{(1,\AT), (1,\PN), (0,\AT), (0,\NP)\}$.  As the sample size increases from $n = 500$ to $n = 20000$, the bias consistently converges toward zero and the standard error decreases substantially for all four causal estimands.  
\subsection{{Simulation Results:} Strata Classification Accuracy}
\label{sec:case1-classification}

Since the true strata $G$ are observable in simulation studies, we compare the strata classification accuracy between the proposed method and the posterior mode rule.  Figure~\ref{fig:case1-prop-plots} compares the strata  classification accuracy between the proposed method (\texttt{prop}, shown in blue) and the posterior mode approach (\texttt{post}, shown in red) under various combinations of parameters $(\eta, \delta)$, with a fixed sample size of $n = 20000$.  

\begin{figure}[t]
    \centering
    \includegraphics[width=0.9499\linewidth]{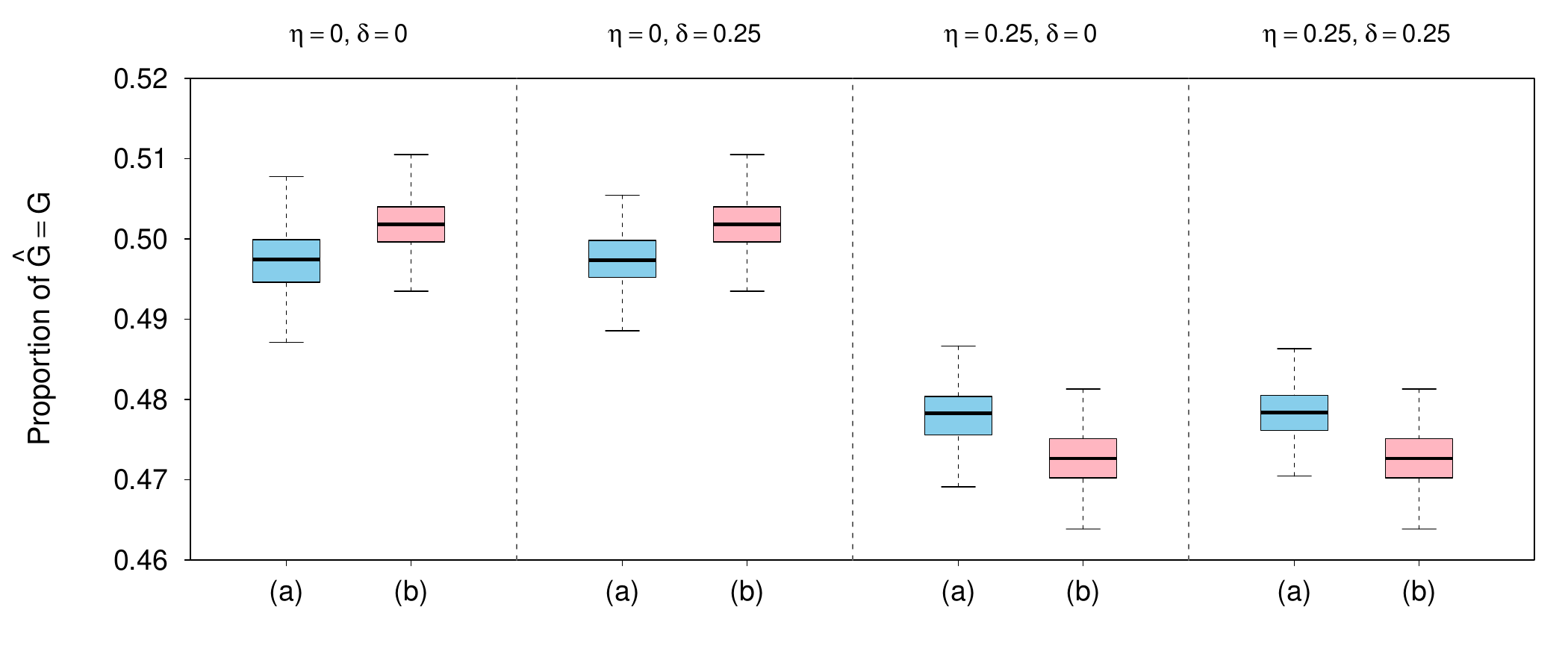}
\caption{Comparison of causal strata classification accuracy under different parameters $(\eta, \delta)$. Blue boxes represent the proposed method, and red boxes represent the posterior mode method.}
    \label{fig:case1-prop-plots}
\end{figure}

\subsection{{Simulation Results:} Policy Revenue}
\label{sec:case1-revenue}

\begin{figure}[t]
    \centering
    \includegraphics[width=0.9499\linewidth]{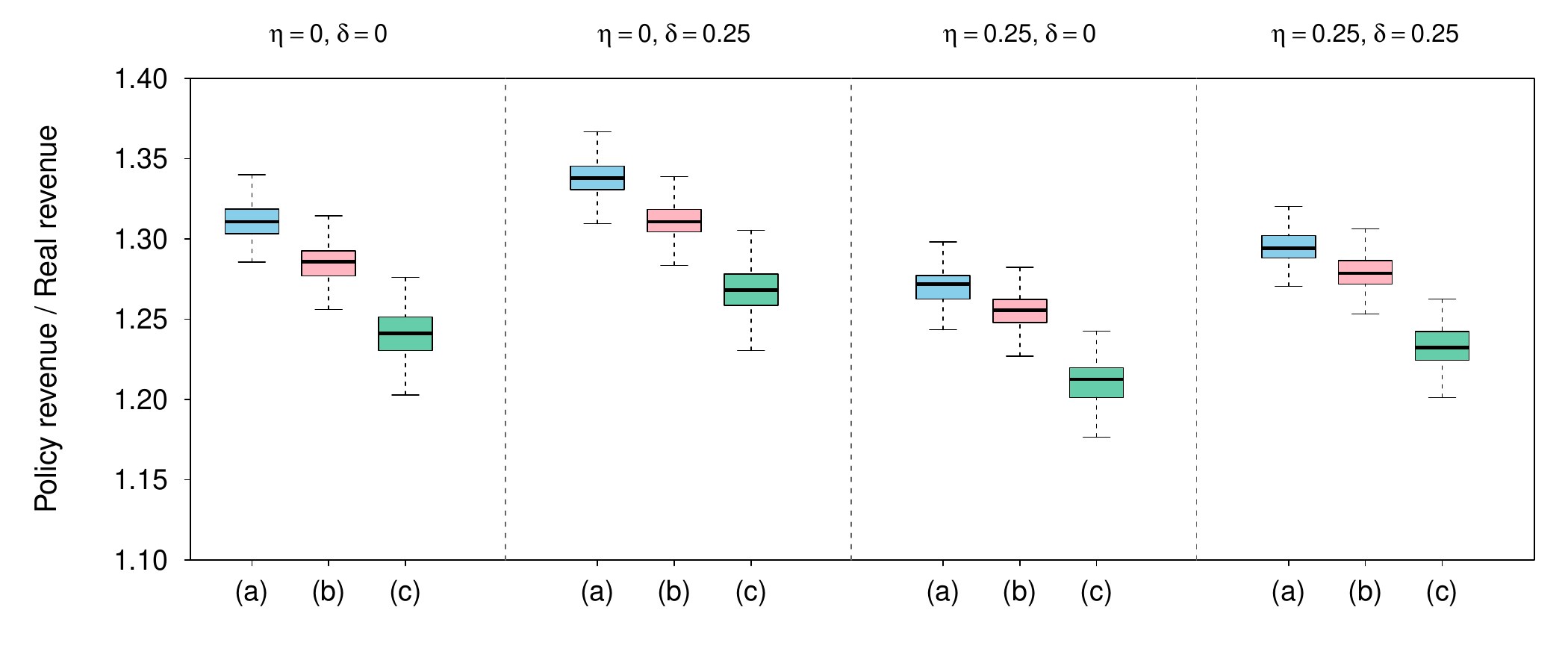}
\caption{Comparison of policy revenue under different parameter configurations ($\eta$, $\delta$). Red boxes represent the proposed method, blue boxes represent the posterior mode method, and green boxes represent the direct optimization method.}
    \label{fig:case1-ER-plots}
\end{figure}

Building upon Section~\ref{sec:case1-classification}, we now evaluate policy learning performance by comparing three approaches: the proposed pseudo-strata learning method (\texttt{prop}), the posterior mode rule (\texttt{post}), and the direct optimization method (\texttt{direct}). The first two follow a two-stage procedure of learning strata then assigning treatments, while the third directly optimizes the policy value function using estimated $\hat{\pi}_{s_0s_1}(X)$ and $\hat{\mathcal{L}}_{z,s_0s_1}(X)$. Throughout all procedures, we set $c_0(X) = c_1(X) \equiv 0$. 
Figure~\ref{fig:case1-ER-plots} compares policy revenue across these three methods under various parameter configurations $(\eta, \delta)$ with sample size $n = 20000$. 

\section{Additional results for real data analysis}
  \label{sec:addition-results-application}
Tables \ref{tab: result-simulation-real-SM-1} and \ref{tab: result-simulation-real-SM-2} provide more detailed data analysis results in Section \ref{app:Point estimation} of the main text.
  \begin{table}[t]
\caption{Simulation studies with bias ($\times 100$), standard error ($\times 100$), and 95\%  confidence interval for the causal estimands.  }
 \label{tab: result-simulation-real-SM-1} 
\centering
\resizebox{0.88049965878\columnwidth}{!}{%
\begin{tabular}{ccccccccc}
\toprule\addlinespace[1mm] 
       &  & \multicolumn{3}{c}{${\mathcal{L}}_{1, \AT}$} &  & \multicolumn{3}{c}{${\mathcal{L}}_{1, \PN}$} \\ \addlinespace[1mm]  \cline{3-5} \cline{7-9} \addlinespace[1mm] 
$\eta$ &  & Point estimate   & SE     & 95\% Confidence Interval              &  & Point estimate  & SE     & 95\% Confidence Interval               \\ \cline{1-1} \cline{3-5} \cline{7-9} \addlinespace[1mm] 
0.00   &  & 143.44           & 58.03  & (29.70, 257.18)   &  & 84.94           & 12.38  & (60.68, 109.20)    \\
0.25   &  & 137.00           & 59.86  & (19.68, 254.32)  &  & 86.42           & 14.19  & (58.61, 114.23)   \\
0.50   &  & 125.96           & 50.72  & (26.54, 225.38)  &  & 89.35           & 14.82  & (60.29, 118.40)    \\
0.75   &  & 131.09           & 50.02  & (33.04, 229.13)  &  & 92.99           & 23.15  & (47.62, 138.36)   \\
1.00   &  & 122.96           & 48.64  & (27.63, 218.30)   &  & 91.76           & 28.17  & (36.55, 146.98)   \\
1.25   &  & 114.70           & 38.17  & (39.89, 189.52)  &  & 86.96           & 29.96  & (28.23, 145.69)   \\
1.50   &  & 119.71           & 30.11  & (60.69, 178.73)  &  & 82.16           & 44.91  & (-5.86, 170.19)   \\
1.75   &  & 115.45           & 23.88  & (68.65, 162.26)  &  & 77.96           & 42.80  & (-5.93, 161.86)   \\
2.00   &  & 105.86           & 18.91  & (68.80, 142.93)   &  & 85.09           & 51.88  & (-16.59, 186.77)  \\
2.25   &  & 98.60            & 15.83  & (67.58, 129.62)  &  & 90.81           & 62.16  & (-31.02, 212.64)  \\
2.50   &  & 91.91            & 13.96  & (64.55, 119.27)  &  & 99.91           & 59.49  & (-16.69, 216.52)  \\
2.75   &  & 87.11            & 12.11  & (63.38, 110.85)  &  & 108.42          & 63.29  & (-15.62, 232.46)  \\
3.00   &  & 82.05            & 11.26  & (59.98, 104.11)  &  & 109.19          & 62.29  & (-12.89, 231.27)  \\
3.25   &  & 78.06            & 11.15  & (56.20, 99.92)    &  & 126.50          & 61.60  & (5.77, 247.23)    \\
3.50   &  & 75.38            & 10.86  & (54.09, 96.67)   &  & 141.58          & 64.40  & (15.36, 267.79)   \\
3.75   &  & 73.54            & 10.80  & (52.37, 94.71)   &  & 139.32          & 67.81  & (6.42, 272.22)    \\
4.00   &  & 73.86            & 8.47   & (57.26, 90.47)   &  & 138.99          & 59.03  & (23.29, 254.69)  \\\toprule
\end{tabular}
}
\end{table}  

  \begin{table}[t]
\caption{Simulation studies with bias ($\times 100$), standard error ($\times 100$), and 95\%  confidence interval for the causal estimands.  }
 \label{tab: result-simulation-real-SM-2} 
\centering
\resizebox{0.88049965878\columnwidth}{!}{%
\begin{tabular}{ccccccccc}
\toprule\addlinespace[1mm] 
       &  & \multicolumn{3}{c}{${\mathcal{L}}_{0, \AT}$} &  & \multicolumn{3}{c}{${\mathcal{L}}_{0, \NP}$} \\ \addlinespace[1mm]  \cline{3-5} \cline{7-9}\addlinespace[1mm] 
$\eta$ &  & Point estimate   & SE     & 95\% Confidence Interval              &  & Point estimate  & SE     & 95\% Confidence Interval               \\  \cline{1-1} \cline{3-5} \cline{7-9} \addlinespace[1mm] 
0.00                 &  & 155.17   & 54.92       & (47.53, 262.80)  &  & 105.26   & 23.96       & (58.29, 152.22) \\
0.25                 &  & 148.44   & 61.98       & (26.94, 269.93) &  & 107.23   & 27.81       & (52.72, 161.75) \\
0.50                 &  & 139.22   & 58.58       & (24.40, 254.04)  &  & 119.43   & 33.59       & (53.59, 185.27) \\
0.75                 &  & 135.16   & 60.76       & (16.08, 254.24) &  & 131.54   & 39.35       & (54.42, 208.67) \\
1.00                 &  & 130.80   & 62.42       & (8.45, 253.15)  &  & 143.71   & 40.62       & (64.10, 223.32)  \\
1.25                 &  & 128.77   & 56.62       & (17.79, 239.76) &  & 144.17   & 46.43       & (53.17, 235.17) \\
1.50                 &  & 123.73   & 46.91       & (31.78, 215.68) &  & 149.24   & 56.22       & (39.06, 259.42) \\
1.75                 &  & 115.10   & 40.10       & (36.52, 193.69) &  & 153.02   & 55.99       & (43.28, 262.75) \\
2.00                 &  & 119.70   & 42.58       & (36.24, 203.16) &  & 149.70   & 61.96       & (28.26, 271.14) \\
2.25                 &  & 114.68   & 37.10       & (41.95, 187.40)  &  & 154.34   & 62.71       & (31.42, 277.26) \\
2.50                 &  & 104.01   & 37.09       & (31.32, 176.70)  &  & 162.19   & 54.55       & (55.27, 269.12) \\
2.75                 &  & 107.66   & 35.07       & (38.92, 176.40)  &  & 165.04   & 57.00       & (53.33, 276.75) \\
3.00                 &  & 104.59   & 30.57       & (44.66, 164.51) &  & 168.11   & 53.69       & (62.88, 273.34) \\
3.25                 &  & 99.45    & 24.95       & (50.55, 148.35) &  & 175.86   & 57.61       & (62.94, 288.78) \\
3.50                 &  & 103.48   & 23.86       & (56.72, 150.25) &  & 165.62   & 55.82       & (56.22, 275.02) \\
3.75                 &  & 105.85   & 25.27       & (56.33, 155.38) &  & 156.37   & 65.76       & (27.49, 285.26) \\
4.00                 &  & 108.06   & 23.23       & (62.53, 153.60)  &  & 172.74   & 59.92       & (55.30, 290.19)  \\\toprule
\end{tabular}
}
\end{table}  

\end{document}